\documentclass{amsart}
\usepackage{amssymb}
\usepackage{mathrsfs}
\usepackage{overpic,color}

\usepackage[numbers]{natbib}

\usepackage[colorlinks,urlcolor=red,citecolor=blue,linkcolor=red]{hyperref}

\newtheorem{theorem}{Theorem}[section]
\newtheorem{proposition}[theorem]{Proposition}
\newtheorem{lemma}[theorem]{Lemma}
\newtheorem{definition}[theorem]{Definition}

\newtheorem{assumption}[theorem]{Assumption}
\newtheorem{remark}[theorem]{Remark}
\newtheorem{example}[theorem]{Example}

\numberwithin{equation}{section}

\newenvironment{enumeratei}
  {\begin{enumerate} }
  {\end{enumerate}}

\newcommand{\half}{\frac{1}{2}}

\DeclareMathOperator{\Var}{Var}

\newcommand{\ep}{\varepsilon}
\renewcommand{\phi}{\varphi}
\newcommand{\de}{\delta}

\newcommand{\Q}{\mathbb Q}
\newcommand{\R}{\mathbb R}
\newcommand{\N}{\mathbb N}

\newcommand{\ind}{1\!\kern-1pt \mathrm{I}}
\newcommand{\rsto}{]\!\kern-1.8pt ]}
\newcommand{\lsto}{[\!\kern-1.7pt [}

\newcommand{\cF}{\mathcal{F}}
\newcommand{\ccL}{\mathscr{L}}

\newcommand{\cT}{\mathcal{T}}

\begin{document}


\title[Bond markets beyond short rate paradigm]{When roll-overs do not qualify as num\'eraire: bond markets beyond short rate paradigms}

\author{Irene Klein, Thorsten Schmidt, Josef Teichmann}





\thanks{The authors thank ETH Foundation for its support of this research project. The first and second author thank the Forschungsinstitut Mathematik at ETH Z\"urich for its generous hospitality.}

\begin{abstract}
We investigate default-free bond markets where the standard relationship between a possibly existing bank account process and the term structure of bond prices is broken, i.e.~the bank account process is not a valid num\'eraire. We argue that this feature is not the exception but rather the rule in bond markets when starting with, e.g., terminal bonds as num\'eraires. 

Our setting are general c\`adl\`ag  processes as bond prices, where we employ directly methods from large financial markets. Moreover, we do not restrict price process to be semimartingales, which allows for example to consider markets driven by fractional Brownian motion. In the core of the article we relate the appropriate no arbitrage assumptions (NAFL), i.e.~no asymptotic free lunch, to the existence of an equivalent local martingale measure with respect to the terminal bond as num\'eraire, and no arbitrage opportunities of the first kind (NAA1) to the existence of a supermartingale deflator, respectively. In all settings we obtain existence of a generalized bank account as a limit of convex combinations of roll-over bonds. 

Additionally  we provide an alternative definition of the concept of a num\'e\-raire, leading to a possibly interesting connection to bubbles. If we can construct a bank account process through roll-overs, we can relate the impossibility of taking the bank account as num\'eraire to liquidity effects. Here we enter endogenously the arena of multiple yield curves. 

The theory is illustrated by several examples.
\end{abstract}

\maketitle

\section{Introduction}\label{section1}

Most of the term structure models in the literature are based on the fundamental assumption that bond prices $ P(t,T) $ together with a num\'eraire bank account process $ B(t) $ form an arbitrage-free market. Formally speaking this means that we can find an equivalent local martingale measure for the collection of stochastic processes $ {{(B(t)}^{-1} P(t,T))}_{0 \leq t \leq T} $ representing bond prices discounted by the bank account's current value. If we assume additionally that those local martingales are indeed martingales, then we arrive at the famous relationship
\begin{equation} \label{bondprice}
P(t,T)=E_Q \big[ \frac{B_t}{B_{T}}| \mathcal{F}_t \big]
\end{equation}
for $ 0 \leq t \leq T $. If we assume alternatively the existence of forward rates, we arrive at the Heath-Jarrow-Morton (HJM) drift condition for the stochastic forward rate process encoding the previous local martingale property (see \cite{HJM}).

On the other hand, bank account processes are limits of roll-over constructions and therefore only approximately given in real markets. Even if there is a bank account process it is not innocent to take it as num\'eraire, since this is an assertion on available liquidity in arbitrary positive and negative multiples of the bank account process. Therefore we avoid assumptions on even the existence of a bank account process and -- in case of existence -- we do not assume that it qualifies as num\'eraire. This can be compared to the famous BGM market model approach, see \cite{bgm:97}, but with the goal in mind to fully characterize the absence of arbitrage. 

It might appear that this relaxation of assumptions leads to more general but possibly less interesting term structure models. However, we show that term structure models, where the bank account process does not qualify as num\'eraire, appear to be a more natural choice leading even to simple interpretations of liquidity effects in the sense of multiple yield curves. In order to avoid ambiguities we also outline in some detail our view on change of num\'eraire, bubbles, and liquidity.

Let us demonstrate our setting with an example lent from Eckhard Platen's benchmark approach. Let $ S^* $ denote a growth optimal portfolio on a finite time horizon $ T^* > 0 $. We claim this portfolio to be a num\'eraire portfolio of our market, see Section \ref{change_of_numeraire}, i.e.~fair prices $ {(\pi_t(X))}_{0 \leq t \leq T} $ of payoffs $ X $ at time $ T > 0 $ should satisfy  the martingale equation
\begin{equation}
\frac{\pi_t(X)}{S_t^*}=E_Q \big[ \frac{X}{S^*_T}| \mathcal{F}_t \big]
\end{equation}
under appropriate integrability assumptions. In particular we obtain bond price processes from the equation
\begin{equation} \label{Sstar}
P(t,T)=E_Q \big[ \frac{S_t^*}{S^*_T}| \mathcal{F}_t \big]
\end{equation}
for $ 0 \leq t \leq T \leq T^*$. If bond prices are bounded by $ 1 $, i.e.~in the case of non-negative interest rates, then the inverse growth optimal portfolio necessarily is a supermartingale, and vice versa. We recall that being a num\'eraire means in our view that arbitrary positive and negative multiples of $ S^\ast $ are available for trading. In this setting \emph{also} the terminal bond $ P(.,T^*) $ qualifies as num\'eraire, i.e.~we can find an equivalent martingale measure $ Q^* $ such that the processes $ \frac{P(.,T)}{P(.,T^*)} $ are $ Q^* $-martingales for $ 0 \leq t \leq T^\ast$. Indeed, define a density process via the properly normalized positive martingale $ {\big(\frac{P(t,T^*)}{S_t^*} \big)}_{0 \leq t \leq T^*} $, which is a martingale by construction, leading to the desired assertion. In other words: the benchmark approach leads to a bona fide bond market without frictions and free lunches.

Consider now the case where the growth optimal portfolio is exogenously given such that $ \frac{1}{S^*} $ is a strict local martingale. Furthermore we assume that the roll-over portfolio leads to the bank account process $ B $ equal to $ 1 $ even though the term structure will be non-trivial due to the strict local martingale property, see Subsection \ref{benchmark-example} for an example of these properties. In contrast the bank account process $ 1 $ does not qualify as num\'eraire, since $ \frac{1}{S^*} $, or equivalently $\frac{1}{P(.,T^\ast)}$ is a strict local martingale. The bond market with respect to the num\'eraire $ P(.,T^*) $ is free of (asymptotic) arbitrage in the classical sense, even if one adds the bank account process as additional traded asset, but we cannot perform a change of num\'eraire towards the num\'eraire $ B = 1 $. 

An interesting aspect of the previous example stems from the introduction of virtual term structures related to bank account processes $ B^n$. We think here of the (finite) roll-over processes, i.e., for a sequence of refining partitions $0=t^n_0<t^n_1<\dots<t_{k_n}^n=T^*$ of $[0,T^*]$ define, for each $n$,
\[
B^n_{t}=\begin{cases}\prod_{i=1}^j\frac{1}{P(t^n_{i-1},t^n_i)}&\text{for $t=t^n_j$, $j=1,\dots,k_n$,}\\
B^n_{t^n_j}&\text{for $t^n_{j-1}<t\leq t^n_j$, $j=1,\dots,k_n$.}\end{cases}
\]
Notice the difference to Definition \ref{rollover}. We have $ \lim_{n \to \infty} B^n=B^\infty = B = 1 $ as announced before, see Subsection \ref{benchmark-example} for a concrete example. These virtual term structures can be interpreted as high-liquidity term structures, which one would actually expect in the market if there was enough liquidity in the respective num\'eraire: this amounts to pricing with the corresponding supermartingale deflator
$$
\frac{1}{E_{Q^*}[\frac{B^n_T}{P(T,T^*)}]} \frac{B^n_T}{P(T,T^*)}  \frac{1}{B^n_T},
$$
which is derived from changing measure by the local martingale density $ \frac{B^n_T}{P(T,T^*)} $. When pricing $ 1 $ at time $ T $ with respect to this deflator we obtain an alternative term structure $ \tilde{P}^n(t,T) $
\begin{align*}
E_{Q^*} \big[  \frac{B^n_T}{P(T,T^*)} \frac{1}{E_{Q^*}[\frac{B^n_T}{P(T,T^*)}]} \frac{1}{B^n_T} \big] = & E_{Q^*} \big[  \frac{1}{P(T,T^*)} \frac{1}{E_{Q^*}[\frac{B^n_T}{P(T,T^*)}]} \big] \\
= & \frac{1}{B^n_0} \tilde{P}^n(0,T) \, ,
\end{align*}
which yields
\begin{equation}
\tilde{P}^n(0,T) = \frac{B^n_0 P(0,T)}{P(0,T^*) E_{Q^*}[\frac{B^n_T}{P(T,T^*)}]} > P(0,T) \, ,
\end{equation}
for each $ n $, i.e., the virtual term structures show lower interest rates (due to higher liquidity) than $ P(t,T) $. In case of $ B^\infty $ we apparently obtain the virtual term structure $ \tilde{P}^\infty(0,T) = 1 $, which corresponds to the highest liquidity virtual term structure, with overnight borrowing at no cost available.

It is the aim of this article to understand bond market dynamics under the weaker assumption that there are no arbitrages with respect to a terminal bond num\'eraire. This is a minimal assumption, which appears to us -- in light of the previous example -- more appropriate. Furthermore we would like to provide an explanation of multiple yield curves from our normative modeling approach: our suggestion in this direction is to look at change of num\'eraire.

From a mathematical point of view we believe that the technology of large financial markets is the right tool to understand the nature of arbitrage in the considered (infinite-dimensional) bond market, since we want to avoid artificially introduced  trading strategies. 
More precisely, we fix a terminal maturity $T^*$ and consider the bond market (for maturities $T\leq T^*$) with respect to the terminal bond as a num\'eraire. It is only allowed to trade in a finite number of assets, but we can take more and more of them and so approximate a portfolio with an infinite number of assets. In contrast to, e.g., \cite{B:DiMa:K:R:1997}, \cite{DeD:P:2005} or \cite{EkelandTaflin2005} we do not introduce infinite--dimensional trading strategies but only approximate by finite portfolios, an idea which stems from the theory of large financial markets. As a direct consequence, we avoid pitfalls for measure-valued strategies pointed out in \cite{Taflin:2011}. A second advantage is that we are able to consider markets driven by general c\`adl\`ag processes with only a weak regularity in maturity. This extends beyond semimartingale models as considered in the above mentioned articles and in \cite{DoeberleinSchweizer2001}.

The structure of the article is as follows: in Section \ref{change_of_numeraire} we introduce some basic ideas to the notion of num\'eraire which will be of importance for the whole article. In Section~\ref{section2} and \ref{section3} we introduce our model for a bond market with an appropriate interpretation as a large financial market and characterize notions of no arbitrage.

In Section~\ref{section4} we relate the appropriate no arbitrage assumption, which is no free lunch (NFL) on the bond market, to the global existence of an   equivalent local martingale measure. Indeed, we can prove the existence of an equivalent local martingale measure for all bonds with maturity $T\leq T^*$  in terms of the  bond $P(t,T^*)$  as num\'eraire. This is in contrast to common bond market models in the literature which often start with the assumption of existence of an equivalent (local) martingale measure, whereas we directly define a notion of no arbitrage and then  the existence of a local martingale measure follows.

In Section~\ref{section6} we prove by a Koml\'os type argument that under the assumption of NAFL there exists a candidate process for the bank account as a limit of convex combinations of roll-over bonds. This bank account is a supermartingale in terms of the terminal bond and therefore the bond market stays free of arbitrage when we add the bank account to the market.

In Section~\ref{section5} we will see that it is possible to further relax the assumptions on the bond market. If we only assume that the bond market does not allow asymptotic arbitrage opportunities of first kind (AA1) in the sense of large financial markets as in \cite{Kab:Kra:1994}, we cannot guarantee  the global existence of an equivalent local martingale measure. However, we can prove that there exists a strictly positive supermartingale deflator for each sequence of bonds with maturities that do not induce an AA1. If there exists a dense sequence of maturities in $[0,T^*]$, such that the induced large financial market is free of AA1, then there exists a supermartingale deflator for the bond market with all maturities in $[0,T^*]$. In this relaxed setting we can still show the existence of a generalized bank account as a limit of convex combinations of roll-over bonds, which is a supermartingale in terms of the terminal bond. This section is related to results of Kostas Kardaras, see, e.g., \cite{Kar}

In Section~\ref{section7} we illustrate the setup with four examples:  first, we consider the  example with optimal growth portfolio being a strict local martingale. Second,  a bond market model driven by fractional Brownian motion is studied, where bond prices are not semimartingales, but bond prices in terms of the num\'eraire are. Third, we consider an extension of the Heath-Jarrow-Morton approach where the bond prices as functions of the maturity are continuous but of unbounded variation such that a short-rate does not exist. Fourth, we give an example illustrating possible pitfalls when considering limits of roll-overs as num\'eraire: in a setting of not uniformly integrable bond prices,  the limit of roll-overs does not qualify as num\'eraire because it reaches zero with probability one.

\section{Change of num\'eraire, liquidity and bubbles}\label{change_of_numeraire}

In this section we outline some basic definitions and conclusions on num\'eraires, liquidity and bubbles, since from the next section changes of num\'eraire will be used frequently. The goal of this section is to add some possibly new definition to the large literature on these issues, however, no deep results are proved.

In seminal works on the absence of arbitrage in financial markets the num\'eraire (portfolio) plays a distinguished r\^ole, see \cite{Delb:Schach:1994}. Additionally in  markets with stochastic interest rates, or foreign exchange markets, change of num\'eraire is an important technique. It turns out that the question which portfolios do qualify as num\'eraire is surprisingly subtle and often only indirectly solved: usually one characterizes possible changes of num\'eraire mathematically but \emph{no economically reasonable} properties of num\'eraire portfolios are laid down (see for instance the seminal work \cite{Delb:Schach:1995}, where num\'eraire portfolios are characterized as maximal, admissible, strictly positive portfolios). We would like to close this small gap in the following paragraphs by providing a simple definition of num\'eraire portfolios, which can be also mirrored in the world of bubbles and liquidity, and which still makes sense in discrete time and under trading constraints.

Intuitively, a portfolio can be used as num\'eraire if it is strictly positive and allows for short-selling, i.e.~the investor is able to find a reasonable counterparty from whom the portfolio can be borrowed and she sells it then on the market. Short-selling might require arbitrarily high credit lines when the portfolio is to be returned, so the counterparty faces the risk of the investor's bankruptcy. Mathematically speaking this might lead to arbitrages in the virtual world after a change of num\'eraire (see \cite{Delb:Schach:1995}). Hence some conditions on the behavior of the short-sold portfolio from below must be imposed. On the other hand we do not want to bound the short-sold portfolio from below by some number, hence the usual admissibility condition is too strong. Instead of admissibility of the short-sold portfolio we require a uniform integrability condition with respect to some equivalent local martingale measure. In other words: we extend the notion of traded portfolios a bit beyond admissibility and call a strictly positive portfolio $N$ a num\'eraire if $N$ and $ -N $ are traded in this extended sense. Such approaches have been successfully investigated in \cite{Delb:Schach:1997} in the context of workable claims, or in \cite{Strasser:2003} via a re-formulation of the Ansel-Stricker framework \cite{A:S:1994}. We consider here the second approach which seems to us slightly more descriptive, but we could also simply formulate everything in the context of workable claims. Notice that the second definition also makes sense under trading constraints.

We give a precise definition which reflects this insight and which \emph{leads} to the well known change of num\'eraire formulas, see \cite{Delb:Schach:1995}. Furthermore we relate this intuitive and economically meaningful definition with the notion of bubbles: a positive portfolio is modeled in a bubble state if it does \emph{not} qualify as num\'eraire. Both concepts will play an important role when it comes to the notion of liquidity in bond markets.

Consider a filtered probability space $(\Omega,{\mathcal F},({\mathcal F}_t)_{t\in[0,T]}, P)$, where the filtration satisfies the usual conditions. The price process of traded assets $(\mathbf{X}_t)_{t \in [0,T]}=(X^0_t,\dots,X^d_t)_{t \in [0,T]}$ is a $d+1$-dimensional adapted process with c\`adl\`ag trajectories, where at least one process, say $X^0 $, is positive, i.e.~$X^0>0$. We introduce the process of discounted assets,
\[
\mathbf{S} := (1, \frac{X^1}{X^0},\dots,\frac{X^d}{X^0}) \,
\]
and assume without loss of generality that we are dealing from now on with a semimartingale $S$. Let $\mathbf{H}$ be a predictable $\mathbf{S}$-integrable process and denote by $(\mathbf{H}\cdot \mathbf{S}) $ the stochastic integral process of $\mathbf{H}$ with respect to $\mathbf{S}$, the \emph{(portfolio) wealth process}. The process $ \mathbf{H} $ is called an \emph{$a$-admissible trading strategy} if there is $a \geq 0$ such that $(\mathbf{H}\cdot \mathbf{S})_t\geq -a$ for all  $t \in [0,T]$. A strategy is called admissible if it is $a$-admissible for some $ a \geq 0 $. Define
\begin{align*}
\mathbf{K} =\{(\mathbf{H}\cdot \mathbf{S})_T:\text{$H$ admissible}\}
\text{ and }
\mathbf{C} &=\{ g \in L^\infty(P): g \le f \text{ for some }f \in K \}.
\end{align*}
Then $ \mathbf{K} $ and $\mathbf{C}$ form convex cones in $L^0(\Omega,\cF,P)$.

The condition \emph{no free lunch with vanishing risk} (NFLVR) is the right concept of no arbitrage, see \cite{Delb:Schach:1994} and \cite{Delb:Schach:1998}.
\begin{definition}
The market $ \mathbf{S} $ satisfies (NFLVR) if
$$ \bar {\mathbf{C}} \cap L_+^\infty(P)=\{0\},$$
where $\bar {\mathbf{C}}$ denotes the closure of $C$ with respect to the norm topology of $L^\infty(P)$.
\end{definition}
This means that a free lunch with vanishing risk exists, if there exists a free lunch $f\in L_+^\infty(P)$, which can be approximated by a sequence of portfolio wealth processes $(f_n)=((\mathbf{H_n} \cdot \mathbf{S}))\in K$  with $\frac{1}{n}$-admissible integrands $\mathbf{H}_n$, such that
\[
\lim_{n \to \infty} \parallel f - f_n \parallel_\infty = 0
\]
with respect to the norm topology of $ L^{\infty}(P) $. Define the set $\mathbf{M}_e$ of equivalent separating measures  as
\begin{align*}
 \mathbf{M}_e  &=\{Q\sim P|_{\cF_T}: E_Q[f]\le0\text{ for all $f\in \mathbf{K}$}\} \, .
\end{align*}
If $\mathbf{S}$ is (locally) bounded then $\mathbf{M}_e$ consists of all equivalent probability measures such that $\mathbf{S}$ is a (local) martingale.

Having a general change of num\'eraire theorem in mind it turns out that the concept of admissibility is too strong, since we want to talk about unbounded portfolio wealth processes and their negative to be admissible. Also we want to consider market extensions of the market $ \mathbf{S} $ by assets $\mathbf{Y} $. We assume from now on (NFLVR) for the market $ \mathbf{S} $. We call assets $ \mathbf{Y} $ a market extension of $ \mathbf{S} $ if $ \mathbf{S}':= (\mathbf{S},\mathbf{Y}) $ satisfies (NFLVR). We define in the sequel a larger class of trading strategies which we call $Q$-admissible. This is a generalization of admissibility as introduced above, i.e.~every admissible strategy is $Q$-admissible. The definition is in spirit of the results of Eva Strasser in \cite{Strasser:2003}.
\begin{definition}
Fix $ Q \in \mathbf{M}_e $. Consider an extension of the original market $ \mathbf{S}':=(\mathbf{S},\mathbf{Y}) $ by finitely many assets $\mathbf{Y}$ such that the process $ \mathbf{S}' $ is a $Q$-local martingale. Consider furthermore a predictable, $\mathbf{S}'$-integrable process $\boldsymbol{\phi}$ and the sequence of hitting times
$$ \sigma_n := \inf\{ t \ge 0: (\boldsymbol{\phi} \cdot \mathbf{S}')_t \le -n\}, \quad n \ge 1. $$
The trading strategy $\boldsymbol{\phi}$ is called $Q$-\emph{admissible} (such as the corresponding stochastic integral, the wealth process), if
\[
\liminf_{n \to \infty} E_Q[ (\boldsymbol{\phi} \cdot \mathbf{S}')_{\sigma_n}^- \ind_{\{\sigma_n < \infty}\}] = 0.
\]
\end{definition}
Define
\begin{align*}
\mathbf{L}^{Q} = \{ x + (\boldsymbol{\phi} \cdot \mathbf{S}') : x \in \R, \quad \boldsymbol{\phi} \text{ is $Q$-admissible} \}.
\end{align*}
and
\begin{align*}
\mathbf{L} = \cup_{Q \in \mathbf{M}_e} \mathbf{L}^Q \, .
\end{align*}

\begin{remark}
We extend the set of admissible portfolios but due to Theorem 3 in \cite{Strasser:2003} we do not introduce arbitrages, since every wealth process $ (\boldsymbol{\phi} \cdot \mathbf{S}') $ for a $Q$-admissible strategy is a supermartingale. We also do not introduce free lunches, since this notion \emph{only} depends on $a$-admissible strategies.
\end{remark}

\begin{remark}
We could use a less general but more appealing definition of $L^Q$ when we do not allow for a market extension $ \mathbf{S}' $. Then num\'eraires are traded portfolios in the original market $ \mathbf{S} $. In our definition all possible price processes for payoffs at time $ T $ are added. Notice that we should consider $L^Q$ as set of trading strategies of our market, but not their union, since the union might contain contradictory pricing structures for one payoff.
\end{remark}

Now we are in the position to make our intuitive definition of num\'eraire portfolios precise: a num\'eraire portfolio is a strictly positive portfolio which allows for short-selling, i.e.~the negative of its wealth process is still given by a $Q$-admissible trading strategy for some $ Q \in \mathbf{M}_e $, and hence is an element of $\mathbf{L}$.
\begin{definition}\label{def:numeraire}
A strictly positive process $N \in \mathbf{L} $ with $N_0=1$ is called a \emph{strong num\'eraire} (in discounted terms with respect to $S^0$), if
\begin{align}\label{eqdef:numeraire}
 N \in \mathbf{L}^Q\text{ and } - N \in \mathbf{L}^Q
\end{align}
for all $Q\in \mathbf{M}_e $. It is called \emph{weak num\'eraire} (in discounted terms with respect to $S^0$), if \eqref{eqdef:numeraire} holds for at least one $Q\in \mathbf{M}_e$, i.e.~$N$ and $-N$ are elements of $\, \mathbf{L} $.
\end{definition}

This definition has a clear economic meaning and easy consequences: as it should be, a weak num\'eraire qualifies as an accounting unit, where the classical change of num\'eraire technique is possible: there exist an equivalent measure $ Q \in \mathbf{M}_e $ under which $N=(1+(\boldsymbol{\phi} \cdot \mathbf{S}'))$ is a true $Q$-martingale.
\begin{theorem}
The following statements are equivalent:
\begin{enumeratei}
\item A strictly positive process $N$ with $ N_0 = 1 $ is a weak num\'eraire.
\item There exists $Q\in \mathbf{M}_e$ such that $ N $ is a $Q$-martingale.
\end{enumeratei}
\end{theorem}
\begin{proof}
Both directions are easy: if there exists $ Q \in \mathbf{M}_e $ such that $ N $ is a true $Q$-martingale, then by adding $ N $ to the market  $ \mathbf{S} $ we obtain an element of $\mathbf{L}^Q$, but due to its uniform integrability $-N \in \mathbf{L}^Q $: hence $N$ is a weak num\'eraire. If, on the other hand, $N \in \mathbf{L}^Q $ for some $ Q \in \mathbf{M}_e $, then $ N $ is a $Q$-supermartingale \emph{together} with $ - N $, which in turn means that $ N $ is a $Q$-martingale.
\end{proof}

Our definition of a num\'eraire has a clear relation to bubbles: a portfolio or an asset which does not qualify as num\'eraire is in a bubble state. Again this very intuitive definition leads to the meanwhile classical definition of a bubble, see \cite{CoxHobson:2005}. In other words: if an asset $S^i$ is a strict local martingale under any $Q\in \mathbf{M}_e$,  $-S^i$ is not $Q$-admissible and hence it does not qualify as weak num\'eraire.

\begin{definition}\label{def:bubble}
A strictly positive process $B \in \mathbf{L} $ is (modeled) in a \emph{strong bubble state} if $ - B \notin \mathbf{L} $, i.e.~for all $Q \in \mathbf{M}_e $ the wealth process $B$ is a strict local martingale. It is (modeled) in a \emph{weak bubble state} if $ - B \notin \mathbf{L}^Q $ for some $ Q \in \mathbf{M}_e $, i.e.~for this $ Q \in \mathbf{M}_e $ the wealth process $ B $ is a strict local martingale.
\end{definition}

\begin{theorem}
A strictly positive portfolio $ B \in \mathbf{L} $ with $ B_0 = 1 $ is in a strong bubble state if and only if $B$ does not qualify as weak num\'eraire portfolio. A strictly positive portfolio $ B \in \mathbf{L} $ with $ B_0 = 1 $ is in a weak bubble state if and only if $B$ does not qualify as strong num\'eraire.
\end{theorem}

\begin{remark}
Notice that this notion of bubble, such as the notion of num\'eraire, depends crucially on the set of trading strategies, which in turn under constraints also leads to notions of bubbles in discrete time. Conditions classifying certain strict local martingales and the relation to bubbles may be found in \cite{MijatovicUrusov2012}.
\end{remark}

\begin{definition}\label{def:virtual}
Let $ V \in \mathbf{L} $ with $V_0=1$ be a weak bubble, i.e.~there is $Q\in\mathbf{M}_e$ such that $ V $ is a strict $Q$-local martingale. Consider $T\in[0,T^*]$ and define the probability measure $Q^{V_T}$ by
\[
E_{Q^{V_T}}[Y] := \frac{E_Q[YV_{T}]}{E_Q[V_{T}]}
\]
for bounded measurable $ Y $. We call the market discounted by $V$ a \emph{virtual market} and its prices \emph{virtual prices} with respect to $V_T$, i.e.~the price of a discounted (with respect to $S^0$) $\cF_T$-measurable claim $ Y $ in this virtual market with respect to the pricing measure $ Q^{V_T} $ is given by
\[
 E_{Q^{V_T}} \big[\frac{Y}{V_T} \big] = \frac{E_Q[Y]}{E_Q[V_{T}]} \, .
\]
We call the difference of prices
\begin{align}\label{eq:liquidity}
E_{Q^{V_T}} \big[\frac{Y}{V_T} \big] - E_Q[Y] = E_Q[Y] \big( 1 - \frac{1}{E_Q[V_{T}]} \big), \quad 0 \le T \le T^*
\end{align}
the \emph{term structure of (il-)liquidity premia} of the num\'eraire $V$ with respect to the pricing measure $ Q $.
\end{definition}

\begin{remark}
Imagine the following interesting economic situation: a financial institution does business by selling portfolios $V$, which this institution itself models in a bubble state. Selling portfolios means shortening them, so -- even though -- there is a belief in a bubble state in fact the institutions behavior indicates that shortening is possible. From a regularity point of view the institution should be forced to use $V$ as num\'eraire adding illiquidity premia to their internal pricing. As a consequence the institution would stop arbitraging by shortening a bubble since in their \emph{own} pricing system the arbitrage stops being visible. In other words: if some financial institutions live on real estate bubble phenomena they should be forced to use real estate indices for discounting in their risk models.
\end{remark}

\begin{example}
We revisit the bond market example from the introduction: $ S^\ast $ is considered as num\'eraire (in un-discounted terms) from the very beginning, whereas the strict local martingale $ \frac{1}{S^\ast} $ (corresponding to the new num\'raire $1$, discounted in $S^\ast$ terms) does not lead to a change of measure. Even more, the process $ \frac{1}{S^\ast} $ describes at least a weak bubble in this market model. If we still make the change of numeraire we obtain a flat (virtual) term structure looking as if generated by the bank account $ B^\infty = 1 $. However, this arbitrage cannot be realized since we cannot  shorten $ \frac{1}{S^\ast} $ in our market model. Notice that this terminology allows to talk about more than one price for traded products with payoff $1$ at time $ T $.
\end{example}

\begin{example}
The setting of \emph{relative arbitrages} as considered for example in \cite{FernholzKaratzas2005} with respect to a given portfolio $ N $ can also be analyzed from the point of view of num\'eraire portfolios. If $N$ is a weak num\'eraire, then relative arbitrages lead to arbitrages in the original market. If instead $ N $ does not qualify for a weak num\'eraire, then -- even though there are relative arbitrages -- there is still (NFLVR) in the original market possible. Again we can speak of several (virtual) prices for one payoff without being able to realize the arbitrages.
\end{example}

\section{Market models for bond markets}\label{section2}

We consider the following model for a bond market. Let $(\Omega,
\mathcal{F}, (\mathcal{F}_t)_{t\geq0}, P)$ be a filtered
probability space where the filtration satisfies the usual
conditions. For each $T\in[0,\infty)$ we denote by
$(P(t,T))_{0\leq t\leq T}$ the price process of a bond with
maturity $T$. For all $T$, $(P(t,T))_{0\leq t\leq T}$ is a
strictly positive c\`adl\`ag stochastic process adapted to
$(\mathcal{F}_t)_{0\leq t\leq T}$ with $P(T,T)=1$. We assume that the price process is almost surely right continuous in the second variable,
where the nullset does not depend on $t$, indeed we make
\begin{assumption}\label{nullsetass}
There is $N\in\mathcal{F}$ with $P(N)=0$
such that
$$N\supseteq\bigcup_{t\in[0,\infty)}\{\omega: T\to P(t,T)(\omega)\text{ is not right continuous}\}.$$
\end{assumption}

\noindent For a generic process $X$ and a stopping time $\tau$
we denote by $(X^\tau_t)=(X_{t \wedge\tau})$ the process stopped at $\tau$.

\begin{assumption}\label{unlocbd} We make the following assumption on uniform local boundedness for
$P(.,T)$ and local boundedness for $P(.,T)^{-1}$:
\begin{enumerate}
\item[1)] For any $T$ there
is $\ep>0$, an increasing sequence of stopping times
$\tau_n\to\infty$ and $\kappa_n\in [0,\infty)$ such that
$$P(t,U)^{\tau_n}\leq \kappa_n,$$ for all $U\in[T,T+\ep)$ and all
$t\leq T$.
\item[2)]
There exists a nonempty set $\cT\subset [0,\infty)$  such that
$\left(\frac1{P(t,T^*)}\right)_{0\leq t\leq T^*}$ is locally bounded for all $T^*\in \cT$ .
\end{enumerate}
\end{assumption}

\noindent The set $\cT$ denotes the maturities of those bonds which we shall consider as candidate num\'eraires.

\begin{remark}
Note that Assumption~\ref{unlocbd}  is fulfilled in
the reasonable special case, where $P(.,T)$ and $P(.,T)^{-1}$ are locally
bounded for any $T$ and, for any fixed $t$, the function $T\mapsto P(t,T)$ is
non--increasing. This, for example, holds, if there exists a non-negative short
rate.
\end{remark}

\noindent In the following assumption we consider a num\'eraire related to a \emph{terminal maturity} $ T^* \in \mathcal{T} $.

\begin{assumption}\label{numemm}
For all finite collections of maturities $T_1<T_2<\dots<T_n \leq T^*$ with $ T^* \in \mathcal{T} $ there
exists a measure $Q \sim P|_{\mathcal{F}_{T^*}}$ such that $\left(\frac{P(t,T_i)}{P(t,T^*)}\right)_{0\leq t\leq T_i}$ is a
local $Q $-martingale, $i=1,\dots,n$.
\end{assumption}

\noindent The measure $Q$ from Assumption~\ref{numemm} is called the $T^*$-forward-measure for the finite market consisting of bonds $ P(.,T_i) $, $i=1,\ldots,n$ and the num\'eraire $ P(.,T^*) $.

\begin{remark}\label{remark}
Note that we do {not} assume the existence of a short rate or even a bank account.
Moreover, we do {not} assume that $P(.,T)$ is a semimartingale.
However, Assumption~(\ref{numemm}) implies that, for a finite
collection of maturities, only bonds in terms of the num\'eraire $P(.,T_n)$
are semimartingales under the objective measure $P$, because they are local martingales
under the equivalent measure $Q$. Moreover, they are {locally bounded} because we assumed
that $P(.,T)$ is locally bounded, for any $T$, and $P(.,T^*)^{-1}$ is locally bounded for
$T^*\in \cT$.
\end{remark}

\noindent If there exists a short rate and an equivalent martingale measure for all discounted bond
processes, then Assumption~(\ref{numemm}) follows immediately.

\begin{lemma}\label{lem:2.4}
Assume that there exists the locally integrable short rate process ${(r_t)}_{t \geq 0}$ and let
$B_t:=e^{\int_0^tr_sds}$ for $ t \geq 0 $. Assume that there exists a measure $Q$ such
that $Q|_{\mathcal{F}_t} \sim P|_{\mathcal{F}_t}$ for $t\geq0$ and such
that $\left(B_t^{-1}\,P(t,T)\right)_{0\leq t\leq T}$ is a
$Q$--martingale, for all $T\in[0,\infty)$. Then, for any finite
collection of maturities $T_1<\dots <T_n$, the measure $Q^{T_n}$ with
$$Z^n:=\frac{dQ^{T_n}}{dQ|_{\mathcal{F}_{T_n}}}=\frac{(B_{T_n})^{-1}}{E_Q[B_{T_n}^{-1}]}$$
fulfills Assumption~(\ref{numemm}).
\end{lemma}

\begin{proof}
Let $Q^{T_n}$ be defined as above. We have to show that
$\left(\frac{P(t,T_i)}{P(t,T_n)}\right)_{0\leq t\leq T_i}$ is a
(local) $Q^{T_n}$-martingale, which is the case iff
$\left(\frac{P(t,T_i)}{P(t,T_n)}\cdot E_Q[Z^n|\mathcal{F}_t]\right)_{0\leq t\leq T_i}$ is a (local)
$Q$-martingale. Let $\alpha=\tfrac1{E_Q[B_{T_n}^{-1}]}$, so
$Z^n=\alpha\cdot B_{T_n}^{-1}$. As $(r_t)$ is the short rate and
$Q$ a martingale measure, we can write
$P(t,T_i)=E_Q[\frac{B_t}{B_{T_i}}| \mathcal{F}_t]$ for
$i=1,\dots,n$ and so

$$
\frac{P(t,T_i)}{P(t,T_n)}E_Q[Z^n|\mathcal{F}_t]=  \alpha
\frac{E_Q[B_{T_i}^{-1}|\mathcal{F}_t]}{E_Q[B_{T_n}^{-1}|\mathcal{F}_t]}
E_Q[B_{T_n}^{-1}|\mathcal{F}_t]=
\alpha
E_Q[B_{T_i}^{-1}|\mathcal{F}_t],
$$
which clearly is a $Q$-martingale.
\end{proof}

\section{Bond markets as large financial markets}\label{section3}

Assumption~(\ref{numemm}) means that for a finite selection of bonds considered with respect to a certain
num\'eraire (the bond with the largest maturity) there exists an
equivalent local martingale measure. Our aim will be the
following: for a fixed maturity $T^*\in\cT$, we aim at finding a measure
$Q^*$ such that all bonds with maturity $T\leq T^*$ are local
martingales under $Q^*$ in terms of  the num\'eraire $P(t,T^*)$. In Section~\ref{section4} we will
present a general theorem.

We  introduce a large financial market connected to the bond market.
We start with a short overview of the facts on large financial markets that we will need.
We choose a finite time horizon $T>0$ as this will be sufficient for our purpose.
Let $(\mathbf{S}^n_t)_{t\in[0,T]}$, $n=1,2,\dots$, be a sequence of semimartingales, where $\mathbf{S} ^n$ takes values in $\R^{d(n)}$, based on a filtered probability space $(\Omega,{\mathcal F},({\mathcal F}_t)_{t\in[0,T]}, P)$  where the filtration satisfies the usual assumptions.
For each $n\geq 1$ we define a classical market model (referred to as {\it finite market $n$}) given by the $\R^{d(n)}$--valued semimartingale $\mathbf{S}^n$ (which describes the price processes of $d(n)$ tradable assets). Such classical markets have been treated in  Section \ref{change_of_numeraire}. We assume that the assets are already discounted with respect to a num\'eraire, so we have that one of the $d(n)$ assets equals $1$.  For our purposes it is sufficient to assume that there is a sequence $S^i$, $i=0,1,\dots$ of semimartingales, such that $S^0_t\equiv 1$ and such that  $(\mathbf{S}^n_t)=(S^0_t,S^1_t,\dots,S^n_t)$. In this case $d(n)=n+1$.

Let $\mathbf{H}$ be a predictable $\mathbf{S}^n$-integrable process and, as previously, $(\mathbf{H}\cdot \mathbf{S}^n)_t$ the stochastic integral of $\mathbf{H}$ with respect to $\mathbf{S}$. The process $\mathbf{H}$ is an admissible trading strategy if $\mathbf{H}_0=0$ and there is $a>0$ such that $(\mathbf{H}\cdot \mathbf{S}^n)_t\ge -a$, $0\leq t\leq T$.
Define
\begin{equation}
\mathbf{K}^n =\{(\mathbf{H}\cdot \mathbf{S}^n)_T:\text{$H$ admissible}\}\text{ and }
\mathbf{C}^n =(\mathbf{K}^n-L^0_+)\cap L^{\infty}.\label{K}
\end{equation}
$\mathbf{K}^n$ can be interpreted as the cone of all replicable
claims in the finite market $n$, and $\mathbf{C}^n$ is the cone of all claims in $L^{\infty}$
that can be superreplicated in this market. Define the set $\mathbf{M}_e^n$ of
equivalent separating measures for the finite market $n$ as
\begin{align}\label{Me}
 \mathbf{M}_e^n &=\{Q\sim P: E_Q[f]\le0\text{ for all $f\in \mathbf{C}^n$}\} \\
 &=\{Q\sim P: E_Q[f]\le0\text{ for all $f\in \mathbf{K}^n$}\}.\nonumber
 \end{align}
 If $\mathbf{S}^n$ is (locally) bounded then $\mathbf{M}_e^n$ consists of all
equivalent probability measures such that $\mathbf{S}^n$ is a (local)
martingale.

A {\it large financial market} is the sequence of  the finite market models $n$, i.e.\ the sequence of the market models induced by the $d(n)$-dimensional semimartingales  ${\bf S}^n$. As a consequence, we cannot trade
with an actually infinite number of securities (so that we avoid artificially introduced infinite-dimensional trading strategies), but we can trade in more and more assets and in this way approximate something infinite-dimensional.

We impose the following assumption, which is standard in the theory of large financial markets:
\begin{equation}\label{emm}
\mathbf{M}^n_e\ne\emptyset,\quad\quad\text{ for all $n\in\N$}.
\end{equation}
This implies that any no arbitrage condition (such as {\it no
arbitrage}, {\it no free lunch with vanishing risk}, {\it no free
lunch}) holds for each finite market  $n$.

However, there is still the possibility of various approximations
of an arbitrage profit by trading on the sequence of market models. We will need the notions no asymptotic free lunch and no asymptotic free lunch with bounded risk
and later on no asymptotic arbitrage of first kind, see Section~\ref{section5}.

{\it No asymptotic free lunch} (NAFL) is
the large financial markets analogue of the classical no free
lunch condition (NFL) of Kreps \cite{Kreps}.  We will first recall the classical NFL condition here for a finite market $n$. Let $\mathbf{C}^n$ be defined as in (\ref{K}).

\begin{definition}
The condition NFL holds on the finite market $n$ if
\begin{equation}\label{NFL}
\overline{\mathbf{C}^n}^*\cap L^{\infty}_+ =\{0\} ,
\end{equation}
where $\overline{\mathbf{C}^n}^*$ denotes the weak-star-closure of  $\mathbf{
C}^n$.
\end{definition}

This means by superreplicating claims in an admissible way with a finite number of assets we cannot approximate in a weak-star sense a strictly positive gain.

Now NAFL can be defined in analogous way as the condition NFL but for the whole sequence of sets $\mathbf{C}^n$:

\begin{definition}\label{N(A)FL} A given large financial market satisfies NAFL if
$$
\overline{\bigcup_{n=1}^{\infty}\mathbf{C}^n}^*\cap L^{\infty}_+ =\{0\}.
$$
\end{definition}

If NAFL holds then it is not possible to approximate a strictly positive profit in a weak-star sense by trading in any finite number of the given assets (although we can use more and more of them). Originally the notion NAFL was introduced in \cite{Klein:2000}, see also \cite{Klein:2007}.

\begin{remark}
Note that in the literature the term large financial market is used for a more general concept where each market $n$ is based on a different filtered probability space. So, in our setting, we will not have to deal with technicalities which are common in large financial markets.
\end{remark}
Let us now introduce a large financial markets' structure for the bond market introduced in Section \ref{section2}.

\begin{definition}\label{LFM}
Let $T^*\in\cT$ where $\cT$ is the set from Assumption \ref{unlocbd}.
Fix a sequence $(T_i)_{i\in\N}$ in $[0,T^*]$.
Define the $n+1$-dimensional stochastic process  $(\mathbf S^n)=(S^0,\dots,
S^n)$ on $[0,T^*]$ as follows:
\begin{align}S^i_t=
\begin{cases}
\frac{P(t,T_i)}{P(t,T^*)} &\text{for $0\leq t\leq T_i$}\\
\frac{1}{P(T_i,T^*)}      &\text{for $T_i<t\leq T^*$}
\end{cases}, \label{defSi}
\end{align}
for $i=1,\dots,n$ and $S^0_t=\frac{P(t,T^*)}{P(t,T^*)}\equiv 1$.

The large financial market consists of the  sequence of
classical market  models given by the $(n+1)$-dimensional stochastic processes
$(\mathbf S^n)_{t\in[0,T^*]}$ based on
the filtered probability space
$\left(\Omega,\mathcal{F},(\mathcal{F}_t)_{t\in[0,T^*]},P|_{\mathcal{F}_{T^*}}\right)$.
\end{definition}


\begin{definition}\label{NAFL}
The bond market ${(P(t,T))}_{0 \leq t \leq T}$ for $ 0 \leq T \leq T^*$ satisfies NAFL if there exists a dense sequence
$(T_i)_{i\in\N}$ in $[0,T^*]$, such that  the large financial market of
Definition~\ref{LFM} satisfies the condition NAFL.
\end{definition}

Since all involved semimartingales $\mathbf{S}^n$ are locally bounded due to Assumption \ref{unlocbd}, it is sufficient to deal with
equivalent local martingale measures. Hence, the set $\mathbf M^n_e$ from \eqref{Me} is given as follows:
$$
\mathbf M^n_e=\{Q^n\sim P|_{\mathcal{F}_T}: \mathbf S^n \text{
local $Q^n$-martingale}\}.
$$

By Assumption~\ref{numemm} we have that $\mathbf M^n_e\neq\emptyset$ for all $n\in\mathbb N$, so the standard
assumption (\ref{emm}) for large financial markets holds. Note that this also implies that each $\mathbf{S}^n$ is a semimartingale, so this is not a problem in
Definition~\ref{LFM}.

\begin{remark}\label{weaker_ass}
Obviously Assumption~\ref{numemm} can be weakened if the bond market satisfies condition NAFL.
Indeed it is sufficient to assume that all the processes from Definition~\ref{NAFL}
$(S^n_t)_{0\leq t\leq T^*}$, $n\in\N$, are semimartingales. Then the stochastic
integrals and therefore the sets $\mathbf{C}^n$ make sense.
In this case, the existence of a local martingale measure  for a finite number
of assets follows by NAFL
 as $\overline{\bigcup_{n=1}^{\infty}\mathbf{C}^n}^*\cap L^{\infty}_+ =\{0\}$
implies that
$\overline{\mathbf{C}^n}^*\cap L^{\infty}_+ =\{0\}$, for all $n$.
Hence NFL holds for all $(S^0,\dots,S^n)$, all $n$ and so Assumption~\ref{numemm} follows.
\end{remark}

For continuous price processes the NAFL-condition  can be replaced by the less technical but intuitively
reasonable condition {\it no asymptotic free lunch with bounded risk} NAFLBR which we state here.

\begin{definition}
 On the large financial market there
is an asymptotic free lunch with bounded risk AFLBR if there are
$\alpha>0$ and $f^{k}\in\mathbf K^{n_k}$ such that
\begin{enumerate}
\item[(i)] $f^{k}$ comes from a $1$-admissible integrand,
\item[(ii)] $P^{n_k}(f^{k}\ge \alpha)\ge\alpha$ for all $k\in\N$ and,
\item[(iii)] $\lim_{k\to\infty} P^{n_k}(f^{k}<-\ep)=0$
for all $\ep>0$.
\end{enumerate}
The condition NAFLBR holds if there does not exist an AFLBR.
\end{definition}

If there is an asymptotic free lunch with bounded risk, then, with
$0$ initial investment, it is possible to approximate a positive
profit by trading on a subsequence of market models. The losses
stay bounded below by $-1$ and even tend to $0$ in probability.

\section{Global existence of an equivalent local martingale measure}\label{section4}

The large financial market induced by the bond market provides an adequate framework to analyze  existence of an
equivalent local martingale measure. For each $T^*\in\cT$ with the set $\cT$ from Assumption~\ref{unlocbd},
 we will find a measure $Q^*$ such that all bond prices  with maturity $T\leq T^*$ discounted by the num\'eraire $P(.,T^*)$, are local
martingales under $Q^*$.

In fact, we immediately obtain a measure $Q^*\sim P|_{\mathcal{F}_{T^*}}$ such that
$\left(\frac{P(t,T_i)}{P(t,T^*)}\right)_{0\leq t\leq  T_i}$
is a local $Q^*$-martingale for all $ T_i $ in the dense subset of maturities of Definition~\ref{NAFL}.
This is just the classical Kreps-Yan result which we state in an abstract version below, for a proof see \cite{Schach:1994}.
It remains to show that the local martingale-property holds for all maturities $ T\in[0,T^*]$.

\begin{theorem}[Kreps, Yan]\label{KY} Let $C$ be a convex cone in $L^{\infty}$ such that $-L^{\infty}_+\subseteq C$, $C$ is weak-star-closed and $C\cap L_+^{\infty}=\{0\}$. Then there exists $g$ in $L^1$ such that $g>0$ a.s.~and $E[fg]\leq 0$ for all $f\in C$.
\end{theorem}

\begin{theorem}\label{th1}
Fix any $T^*\in\cT$ and let Assumption~\ref{nullsetass}, \ref{unlocbd} and
Assumption~\ref{numemm} hold.
Then, the bond market satisfies NAFL (see Definition~\ref{NAFL}), if and only if
there exists
a measure $Q^*\sim P|_{\mathcal{F}_{T^*}}$ such that
$\left(\frac{P(t,T)}{P(t,T^*)}\right)_{0\leq t\leq T}$
is a local $Q^*$-martingale for all $ T\in[0,T^*]$.
\end{theorem}

\begin{remark}
In Theorem~\ref{th1} we consider the NAFL condition for the large financial market as in Definition~\ref{LFM} with respect to one fixed, dense sequence of $T_i$ in $[0,T^*]$. However, as there is a local martingale measure for all bond prices discounted by the num\'eraire, the general theorem about NAFL in large financial markets implies that for any large financial market (induced by the bond market via any sequence of maturities) NAFL holds. In particular NAFL follows from the existence of $Q^*$ by \cite{Klein:2000}, \cite{Klein:2007}.
\end{remark}

\begin{proof}[Proof of Theorem~\ref{th1}]
We denote by $(T_i)_{i\in\mathbb{N}}$ the dense sequence from Definition~\ref{NAFL}.  Consider
the large financial market of Definition~\ref{LFM}. By Theorem~\ref{KY} we get for $C=\overline{\bigcup_{n=1}^{\infty}\mathbf{C}^n}^*$ a $g\in L^1(\Omega, \mathcal{F}_{T^*},P)$, $g>0$
such that $E[fg]\leq 0$ for all $f\in C$. Take $\frac{g}{E[g]}$ as the density of a probability measure $Q^*\sim P|_{\mathcal{F}_{T^*}}$. As all
$S^i=\frac{P(t,T_i)}{P(t,T^*)}$ are locally bounded this gives that $S^i$ is a local $Q^*$-martingale.
Indeed, choose $\tau$ such that $(S^i_{t\wedge\tau})_{0\leq t\leq T_i}$ is bounded, and let $s<t\leq T_i$, $A\in\mathcal{F}_s$. Then we have that
$\pm(\ind_{\rsto0,\tau\rsto}\ind_A\ind_{]s, t]}\cdot S^i)_T=\pm\ind_A(S^i_{t\wedge\tau}-S^i_{s\wedge\tau})\in\mathbf{C}^i$. This gives the local martingale property under $Q^*$.

It remains to show that for any $T< T^*$ which is not an element of the sequence $(T_i)$ we get the local martingale
property of $\left(\frac{P(t,T)}{P(t,T^*)}\right)_{0\leq t\leq T}$ with respect to $Q^*$ as well.
As the sequence $(T_i)$ is dense in $[0,T^*]$ there exists a subsequence denoted by $(\tilde T_i)$ with $\tilde T_i\to T$ for $i\to\infty$
(w.l.o.g.~assume that $\tilde T_i\geq  T$ for all $i$).

Let $\frac{P(t,T)}{P(t,T^*)}:= X_t$ and $\frac{P(t,\tilde T_i)}{P(t,T^*)}:=X^i_t$ for each $i$. As a consequence of Assumption~\ref{unlocbd} there exists $\ep>0$, an increasing sequence of stopping
times $\sigma_n\to\infty$ and constants $\kappa_n>0$ such that for all $U\in[ T, T+\ep)$ and all $t\in[0,T]$ we have
that $$\left(\frac{P(t,U)}{P(t, T^*)}\right)^{\sigma_n}\leq\kappa_n.$$
\noindent Hence for $i$ large enough, such that $\tilde T_i\in[T,T+\ep)$, say $i\geq i_{\ep}$, we have that
\begin{equation}\label{ub}
X_{t\wedge\sigma_n}^i\leq \kappa_n\quad\text{for all $t\in[0, T]$.}
\end{equation}
By the first part of the proof, for any $i$,
$X^i$ is a local $Q^*$-martingale. So, (\ref{ub}) gives that, for
$i\geq i_{\ep}$, $(X^i)^{\sigma_n}$ is a $Q^*$-martingale (as it is a bounded local martingale).

Fix $\sigma=\sigma_n$. We will now show that, for all $t\in[0,T]$, we have that, for $i\to\infty$
\begin{equation}\label{conv}
X^i_{t\wedge\sigma}\to X_{t\wedge\sigma}\quad\text{a.s.}
\end{equation}

This holds iff $P(t\wedge\sigma, \tilde T_i)\to P(t\wedge\sigma, T)$
a.s. By right-continuity of $U\to P(t,U)$ it is clear that
$P(t,\tilde T_i)\ind_{\{t<\sigma\}}\to P(t, T)\ind_{\{t<\sigma\}}$ a.s.

So it
remains to show that $P(\sigma,\tilde T_i)\ind_{\{t\geq\sigma\}}\to P(\sigma, T)\ind_{\{t\geq\sigma\}}$ a.s.
Take any $\omega\in\Omega\setminus N$, where $N$ is the nullset of Assumption~\ref{nullsetass}, then we have that
$\sigma(\omega)=s$ for some $s\in[0,T]$ and as $\tilde T_i\downarrow T$ we get
$P(s,\tilde T_i)(\omega)\to P(s,  T)(\omega)$,
and hence
$P(\sigma,\tilde T_i)(\omega)\to P(\sigma, T)(\omega)$,
so \eqref{conv} holds.

Let $s<t\leq T$. By \eqref{conv} we have that
$X^i_{t\wedge\sigma}\to X_{t\wedge\sigma}$ a.s. for all $t\in[0, T]$. Hence
\begin{align}
E_{Q^*}[ X_{t\wedge\sigma}|\mathcal{F}_s]=
E_{Q^*}[\lim_{i\to\infty}X^i_{t\wedge\sigma}|\mathcal{F}_s]
=\lim_{i\to\infty}E_{Q^*}[X^i_{t\wedge\sigma}|\mathcal{F}_s]
=\lim_{i\to\infty}X^i_{s\wedge\sigma} = X_{s\wedge\sigma},\nonumber
\end{align}
where the second equality follows by dominated
convergence as by (\ref{ub}) we have that
$0<X^i_{t\wedge\sigma}\leq\kappa$ for all $i\geq i_{\epsilon}$. The
third equality is the martingale property of $(X^i)^{\sigma}$ for  $i\geq i_{\epsilon}$. This
gives that $( X_t^{\sigma})_{0\leq t\leq  T}$ is a
$Q^*$-martingale. As this holds for each $\sigma$ in the localizing
sequence, $( X_t)_{0\leq t\leq  T}$ is a local $Q^*$-martingale.
\end{proof}

\section{Existence of a bank account}\label{section6}

It is possible to obtain a candidate process for the bank account by a limit of rolled over bonds
as we show now.  Throughout this section we assume that all the assumptions of Theorem~\ref{th1} hold.

\begin{definition}\label{rollover}
Let $0=t^n_0<t^n_1<\dots<t_{k_n}^n=T^*$ be a sequence of refining partitions of $[0,T^*]$.
Define, for each $n$, the \emph{roll-over} $B^n$ as follows: $B^n_0=1$ and

$$B^n_{t}=\begin{cases}\prod_{i=1}^j\frac{1}{P(t^n_{i-1},t^n_i)}&\text{for $t=t^n_j$, $j=1,\dots,k_n$}\\
B^n_{t^n_j}P(t,t^n_j)&\text{for $t^n_{j-1}<t\leq t^n_j$, $j=1,\dots,k_n$}\end{cases}$$
\end{definition}

The sequence of these roll-overs can be viewed as a replacement of a bank account even without passing to a limit.
This is in the spirit of large financial markets, where one often approximates in a finite way for larger and larger $n$ but one does not actually pass to the limit.

We shall see that one can still pass to the limit by taking convex combinations, which will provide us with the notion of a generalized bank account. First we shall observe some properties of the sequence of roll-overs.

\begin{lemma}\label{selff} There exists a self-financing strategy $\hat{\mathbf{H}}^n_t=(\hat H^1_t,\dots, \hat H^{k_n}_t)$ on the market containing the $k_n$-dimensional asset
$\hat{\mathbf{S}}^n(\cdot)=(P(.,t^n_1),\dots, P(., t^n_{k_n}))$ such that $B^n_t= \langle \hat{\mathbf{H}}_t, \hat{\mathbf{S}}_t \rangle$. Discounted by the num\'eraire  $P(t,t^n_{k_n})=P(t,T^*)$ this gives an admissible strategy $\mathbf{H}^n$ such that $\frac{B^n_t}{P(t,T^*)}=\frac1{P(0,T^*)}+(\mathbf{H}^n\cdot \mathbf{S}^n)_t>0$, where $\mathbf{S}^n$ is the process $\hat{\mathbf{S}}^n$  discounted by the num\'eraire $P(t,T^*)$.
In particular, $\big(\frac{B^n_t}{P(t,T^*)}\big)_{0\leq t\leq T^*}$ is a positive local martingale and hence a supermartingale with respect to the measure $Q^*$ of Theorem~\ref{th1}.
\end{lemma}

\begin{proof}
The strategy $\hat{\mathbf{H}}_t^n$ is given as follows. Fix $j$ and let $t^n_{j-1}<t\leq t^n_j$, then
$$\hat H^i_{t}=\begin{cases} \prod_{l=1}^j\frac{1}{P(t^n_{l-1},t^n_l)}&\text{for $i=j$}\\
0&\text{for $i\neq j$},\end{cases}$$
which is equivalent to
$$
\hat H^i_{t}=\sum_{j=1}^{k_n}B^n_{t^n_j}\ind_{(t^n_{j-1},t^n_j]}(t)\delta_{ij},
$$
which is previsible since $B^n_{t^n_j}\in\mathcal{F}_{t^n_{j-1}}$.

\noindent So we get for $t^n_{j-1}<t\leq t^n_j$
\begin{align}
B^n_t &= \langle \hat{\mathbf{H}}_t^n, \hat{\mathbf{S}}_t^n \rangle \nonumber\\
&=\sum_{i=1}^{k_n}\hat H^i_t\hat S^i_t=\prod_{l=1}^j\frac{1}{P(t^n_{l-1},t^n_l)}\cdot P(t,t^n_j)\nonumber\\
&=B^n_{t^n_j}P(t,t^n_j).\nonumber
\end{align}

\noindent This is self-financing as
$$ \langle \hat{\mathbf{H}}_{t^n_{j-1}}^n, \hat{\mathbf{S}}_{t^n_{j-1}}^n \rangle = \langle \hat{\mathbf{H}}^n_{t}, \hat{\mathbf{S}}^n_{t^n_{j-1}} \rangle $$ for $t^n_{j-1}<t\leq t^n_j$.
Indeed the left hand side equals $B^n_{t^n_{j-1}}P(t^n_{j-1},t^n_{j-1})=B^n_{t^n_{j-1}}$ and the right hand side equals $B^n_{t^n_{j}}P(t^n_{j-1},t^n_j)=B^n_{t^n_{j-1}}$.

After discounting by the num\'eraire $P(t,T^*)$  we have the initial investment
$ \langle \hat{\mathbf{H}}_1, {\mathbf{S}}_0 \rangle =\frac{1}{P(0,t^n_1)}\frac{P(0,t^n_1)}{P(0,T^*)}=\frac1{P(0,T^*)}$.
For $t^n_{j-1}<t\leq t^n_j$ as $B^n_{t^n_{j-1}}= \langle \hat{\mathbf{H}}_{t^n_{j-1}},\hat{\mathbf{S}}_{t^n_{j-1}} \rangle = \langle \hat{\mathbf{H}}_{t},\hat{\mathbf{S}}_{t^n_{j-1}} \rangle $ the increment equals
\begin{align}
\frac{B^n_t}{P(t,{T^*})}-\frac{B^n_{t^n_{j-1}}}{P(t^n_{j-1},{T^*})}&=\sum_{i=1}^{k_n}\hat H^i_t\frac{P(t,t^n_i)}{P(t,{T^*})}-\sum_{i=1}^{k_n}\hat{H}^i_t\frac{P(t^n_{j-1},t^n_i)}{P(t^n_{j-1},{T^*})}\nonumber\\
&=\hat H^{k_n}_t(1-1)+\sum_{i=1}^{k_n-1}\hat H^i_t\left(\frac{P(t,t^n_i)}{P(t,{T^*})}-\frac{P(t^n_{j-1},t^n_i)}{P(t^n_{j-1},{T^*})}\right)\nonumber\\
&=B^n_{t^n_j}\left(\frac{P(t,t^n_j)}{P(t,{T^*})}-\frac{P(t^n_{j-1},t^n_j)}{P(t^n_{j-1},{T^*})}\right).\nonumber
\end{align}
Summing the increments up we arrive at the stochastic integral
$$
\frac{B^n_t}{P(t,{T^*})}=\frac1{P(0,{T^*})}+\sum_{j=1}^{k_n}B^n_{t^n_j}\left(\frac{P(t\wedge t^n_j,t^n_j)}{P(t\wedge t^n_j,{T^*})}-\frac{P(t\wedge t^n_{j-1},t^n_j)}{P(t\wedge t^n_{j-1},{T^*})}\right).
$$
As $\frac{B^n_t}{P(t,{T^*})}=\frac{1}{P(0,{T^*})}+\left(\mathbf{H}^n\cdot\mathbf{S}^n\right)_t$ is bounded from below and $\mathbf{S}^n$ is a local $Q^*$-martingale the discounted roll-over is a $Q^*$-supermartingale.
\end{proof}

The existence of limits for refined roll-overs is apparently delicate. The following theorem is proved by a Koml\'os-type argument as in \cite[Lemma~5.2]{FK} and provides us with a generalized bank account, that always exists (under the assumptions of Theorem~\ref{th1}) and which is always a supermartingale with respect to the measure $Q^*$ of Theorem~\ref{th1}.

\begin{theorem}\label{bank_account}
Let $((B^n_t)_{0\leq t\leq T^*})$ be the sequence of roll-overs given as in Definition~\ref{rollover}. There exists a sequence of convex combinations $\tilde{B}^n\in\text{conv}(B^n,B^{n+1},\dots)$ and a c\`adl\`ag stochastic process $(B_t)_{0\leq t\leq T^*}$, henceforward called \emph{generalized bank account}, such that
$$
B_t=\lim_{q\downarrow t}\lim_{n\to\infty}\tilde{B}^n_q,
$$
with $B_0\leq1$ and $0\leq B_t<\infty$, for all $t\leq T^*$. The generalized bank account has the following properties.
\begin{enumerate}
\item The process $(V_t)_{0\leq t\leq T^*}$, where $V_t=\frac{B_t}{P(t,T^*)}$, is a supermartingale with respect to the measure $Q^*$ of Theorem~\ref{th1}.
\item If $0<P(t, T)\leq1$, for all $ T\leq T^*$, then $P(B_t\ge1)=1$, for all $t\leq T^*$.
\end{enumerate}
\end{theorem}

\begin{remark}
In general, we can only say that the process $(V_t)_{0\leq t\leq T}$ is a supermartingale with respect to $Q^*$
(and not a local martingale).
\end{remark}

\begin{proof}
Consider the sequence of roll-overs $M^n_t:=\frac{B^n_t}{P(t,T^*)}$ discounted by the num\-\'eraire $P(.,T^*)$. By Lemma~\ref{selff} these processes are supermartingales (and bounded from below by 0) with respect to the measure $Q^*$ of Theorem~\ref{th1}. The existence of a limit of convex combinations of the $M^n$ follows by Lemma~5.2 of \cite{FK}, we recall the proof here.
Let $\mathcal{D}=\left([0,T^*]\cap\Q\right)\cup\{T^*\}$. This is a dense countable subset of $[0,T^*]$.
By Lemma~A.1.1 of \cite{Delb:Schach:1994} and a diagonalization procedure we find a sequence of processes
$\tilde{M}^n\in\text{conv}\left(\frac{B^n}{P(.,T^*)}, \frac{B^{n+1}}{P(.,T^*)},\dots\right)$ such that,
for all $q\in\mathcal{D}$, $\tilde M^n_q$ a.s. converges to a random variable $V'_q$ with values in $[0,\infty]$.
For each $q$, we have that $\text{conv}(M^n_q,M^{n+1}_q,\dots)$ is bounded in $L^0$, as all $M^n$ are positive supermartingales with starting value $\frac1{P(0,T^*)}$.
Hence for each $\tilde{M} \in \text{conv}(M^n_q,M^{n+1}_q,\dots)$ we have that $E_{Q^*}[|\tilde{M}|]=E_{Q^*}[\tilde{M}]\leq\frac{1}{P(0,T^*)}$,
so the set of convex combinations is bounded in $L^1(Q^*)$ hence in $L^0$. Lemma~A.1.1 of \cite{Delb:Schach:1994} gives then that $V'_q<\infty$ a.s.

Moreover, for $r<q$, $r,q\in\mathcal{D}$, by Fatou and the supermartingale property of $\tilde{M}^n$ we have that
$$
E_{Q^*}[V'_q|\mathcal{F}_r] =E_{Q^*}[\lim_{n\to\infty}\tilde{M}^n_q|\mathcal{F}_r]
\leq\liminf_{n\to\infty}E_{Q^*}[\tilde{M}^n_q|\mathcal{F}_r]
\leq\lim_{n\to\infty}\tilde{M}^n_r=V'_r.
$$

Therefore $(V'_q)_{q\in\mathcal{D}}$ is a discrete $Q^*$-supermartingale. By standard arguments
(using Doob's Upcrossing Lemma) we get that $(V_t)_{0\leq t\leq T^*}$ is a c\`adl\`ag supermartingale,
where, for all $t\in[0,T^*[$,
$$V_t:=\lim_{q \downarrow t}V'_q,$$
and $V_{T^*}:=V'_{T^*}$ (recall that $T^*\in\mathcal{D}$). Note that $V_0\leq\frac1{P(0,T^*)}$ as
$$
V_0=\lim_{q\downarrow0}V'_q =E_{Q^*}[\lim_{q\downarrow0}V'_q | \mathcal{F}_0]\leq V_0'=\frac1{P(0,T^*)}.
$$
Define now $B_t:=V_tP(t,T^*)$, this is a c\`adl\`ag process as $V_t$ and $P(t,T^*)$ are c\`adl\`ag.
As the process $P(t,T^*)$ is right-continuous in $t$ easy computations show that
$B_t=\lim_{q\downarrow t}\lim_{n\to\infty}\tilde{B}^n_q$, where $\tilde{B}^n_q=\sum_{i=1}^{k_n}\lambda^n_iB^i_q=P(q,T^*)\tilde{M}^n_q$.
By definition $B_0=P(0,T^*)V_0\leq 1$.

Let now $P(t,T)\leq 1$, for all $T\leq T^*$, $t\leq  T$. Then we see from the definition of the roll-over as product of terms of the form $\frac{1}{P(t_i,t_{i+1})}\geq1$
that $B^n_t\geq1$ for all $n$,$t$. The same holds for all convex combinations and therefore for the limits as above.
\end{proof}

\begin{remark}
With a view what it means to be a num\'eraire (see Section~\ref{change_of_numeraire}) we can ask why just terminal bonds qualify as num\'eraires by default in our setting: the answer is that we could take any other reasonably behaved stochastic process (the inverse has to be locally bounded) and plug it into Assumption~\ref{numemm} instead of $ P(.,T^*) $. Conclusions would remain the same, of course with a different meaning on what we would call num\'eraire in this setting. For instance we could think of taking discrete roll-over bonds as num\'eraires if we want to claim that this portfolio can be shortened.
\end{remark}

\section{On the existence of a supermartingale deflator and a generalized bank account}\label{section5}

In this section,  we relax the assumptions on the bond market and investigate under which conditions there is a supermartingale deflator. This is motivated by the fact that we are lead to supermartingale deflators
by the very structure of bond market models. Indeed if we have a non-vanishing generalized bank account and decide to choose it as market num\'eraire, see Theorem~\ref{bank_account}, then our theory only provides us with a supermartingale deflator structure.

The results about supermartingale deflators in this section are related to results of Kostas Kardaras, see, e.g., \cite{Kar}. Consider a   large financial market induced by a sequence of semimartingales $S^i$, $i=0,1,\dots$ on a fixed filtered probability space, where the filtration satisfies the usual conditions, such that  $(\mathbf{S}^n_t)_{t\in [0,T^*]}=(S^0_t,S^1_t,\dots,S^n_t)_{t\in[0,T^*]}$.  Recall that  $S^0_t\equiv1$, i.e.~the num\'eraire has been fixed and prices are discounted by the chosen num\'eraire. The sets $\mathbf{K}^n$, $\mathbf{C}^n$, $\mathbf{M}^n_e$ are defined as previously. We assume that each finite market satisfies (NFLVR), i.e.~ \eqref{emm} holds for all $n$. In contrast to the previous sections, we do not assume here that the semimartingales are locally bounded. In this case, the set $\mathbf M^n_e$ as in \eqref{Me} consists of all equivalent probability measures $Q$ such that  stochastic integrals $(\mathbf{H}^n\cdot \mathbf{S}^n)_t$, $0\leq t\leq T^*$, with admissible integrands $\mathbf{H}^n$ (i.e. $(\mathbf{H}^n\cdot \mathbf{S}^n)_{T^*}\in\mathbf{K}^n$) are $Q$-supermartingales. It was shown in \cite{Delb:Schach:1998} that under the condition no free lunch with vanishing risk the set of equivalent sigma-martingale measures for $\mathbf{S}^n$ is dense in the set $\mathbf{M}_e^n$.

The notion {\it no asymptotic arbitrage of first kind} (NAA1)  was introduced in \cite{Kab:Kra:1994}.
\begin{definition}\label{aa1}
A large financial market admits an asymptotic arbitrage opportunity of  first kind if there exists a subsequence, again denoted by $n$, and  trading strategies $\mathbf{H}^{n}$ with
\begin{enumerate}
\item $(\mathbf{H}^{n}\cdot \mathbf{S}^{n})_t\geq -\ep_n$ for all $t\in[0,T^*]$,
\item $ P((\mathbf{H}^{n}\cdot \mathbf{S}^{n})_{T^*}\geq C_n)\ge \alpha,$
\end{enumerate}
for all $ n $, where $\alpha>0$, $\ep_n\to0$ and $C_n\to\infty$.

We say that the large financial market satisfies the condition NAA1 if there are no asymptotic arbitrage opportunities of first kind.
\end{definition}

The following result for large financial markets provides us with supermartingale deflators for bond markets.

\begin{theorem}\label{supermdefllfm}
Consider the large financial market induced by the sequence of semimartingales $(\mathbf{S}^n_t)_{0\leq t\leq T^*}=(S^0_t,S^1_t,\dots,S^n_t)_{t\in[0,T^*]}$, $n=1,2,\dots$
and assume that (\ref{emm}) holds for all $n$.
Then NAA1 holds, if and only if there exists  a strictly positive supermartingale $(Z_t)_{0\leq t\leq T^*}$ with $Z_0\leq1 $,
such that $(Z_tX_t)_{0\leq t\leq T^*}$ is a supermartingale for all processes $X$ with $X_{T^*}\in\bigcup_{n=1}^{\infty}\mathbf{K}^n$.

\noindent Moreover, if NAA1 holds, then:
\begin{enumerate}
\item
if $S^i_t\geq -a$, $0< t\leq T^*$, for some $i\in \N$ and some $a>0$ and $S^i_0\geq0$, then $(Z_tS_t^{i})_{0\leq t\leq T^*}$ is a supermartingale.
\item
If  $S^i_0=0$ and $(S^i_t)_{0\leq t\leq T^*}$ is locally bounded for some $i\in\N$, then $(Z_tS_t^{i})_{0\leq t\leq T^*}$ is a  local martingale.
\end{enumerate}

\end{theorem}

The supermartingale $Z$ is called {\it supermartingale deflator for the large financial market}.

In order to prove Theorem~\ref{supermdefllfm} we will use a result from the theory of large financial markets. In this aspect our proof differs from Kardaras' proofs of similar results.

Under the assumption of NAA1 Yuri Kabanov and Dima Kramkov proved Theorem~\ref{supermdefllfm} in \cite{Kab:Kra:1994} in the complete setting, the most general result can be found in \cite{Kab:Kra:1998}. Note that the general theorem in \cite{KS} was only proved under local boundedness assumptions on all processes. We choose to take the setting of \cite{KS} for our convenience, since its formulation and proof hold in the general case as well.

\begin{theorem}\label{aa1result}
A large financial market satisfies NAA1 if and only if there exists a sequence of probability measures $Q^n\in\mathbf{M}^n_e$ such that $(P^n)\lhd(Q^n)$.
\end{theorem}

$(P^n)\lhd(Q^n)$ means that the sequence of measures $(P^n)$ is contiguous with respect to $(Q^n)$, i.e.~whenever for a sequence of measurable sets $A^n$  we have that $Q^n(A^n)\to0$ then $P^n(A^n)\to0$ (where in our case, of course, each element of the sequence of measures $P^n\equiv P$, all $n$). In the case where, for each $n$, $P^n\ll Q^n$, the notion of contiguity can be interpreted as a uniform absolute continuity in the following sense: for each $\ep>0$ there is $\de>0$ such that, for all $n$ and $A^n\in\mathcal{F}$ with $Q^n(A^n)<\de$ we have that $P^n(A^n)<\ep$.

Let us now proceed with the proof of Theorem~\ref{supermdefllfm}. In order to apply Theorem~\ref{aa1result} we need the following useful lemma.

\begin{lemma}\label{komlos_again}
Let $Q^n$, $n\ge1$, be a sequence of probability measures and $P$ a probability measure on $(\Omega,\mathcal{F}_{T^*})$  such that $Q^n\sim P$, for all $n$, and  $P\lhd(Q^n)$.
Let $Z^n_t=E\big[\frac{dQ^n}{dP}|\mathcal{F}_t\big]$.
Then there exists a c\`adl\`ag supermartingale $(Z_t)_{0\leq t\leq T^*}$ with $Z_0\leq1$ and a sequence of
$\tilde{Z}^n\in \text{conv}(Z^n,Z^{n+1},\dots)$ such that, for all $t\in[0,T^*]$,
\begin{align}\label{defZ_again}
Z_t=\lim_{q\downarrow t}\lim_{n\to\infty} \tilde Z^n_q.
\end{align}
Moreover $P(Z_t>0)=1$, for all $t$.
\end{lemma}

\begin{proof}
Let $\mathcal{D}=\left([0,T^*]\cap\Q\right)\cup\{T^*\}$. The processes $(Z^n_t)_{0\leq t\leq T^*}$ are positive martingales with $Z^n_0=1$. As in the proof of Theorem~\ref{bank_account} there exists a sequence $\tilde Z^n\in\text{conv}(Z^n,Z^{n+1},\dots)$ such that
$Z$ as in \eqref{defZ_again} is a c\`adl\`ag supermartingale, with  $0\leq Z_t<\infty$ for $t \in [0,T^*]$.
As $\tilde Z^n_0=1$, for all $n$, we have that $Z_0\leq1$.

It remains to show that, for all $t$, $P(Z_t>0)=1$. We will show that this holds for $T^*$, which implies the statement for all $t\leq T^*$. Indeed $Z_{T^*}>0$
a.s.~implies $E[Z_{T^*}|\mathcal{F}_t]>0$ a.s. By the supermartingale property,
$$Z_t\geq E[Z_{T^*}|\mathcal{F}_t]>0\quad\text{a.s.}$$

Assume now that for $A=\{Z_{T^*}=0\}$ we have that $P(A)=\alpha>0$. As $T^*\in\mathcal D$, we have that $\ind_A\tilde Z^n_{T^*}\to \ind_AZ_{T^*}=0$ a.s. This implies that, for all $\ep>0$,
\begin{equation}\label{to_zero}
P(\ind_A\tilde Z^n_{T^*}>\ep)\to0.
\end{equation}
Hence, for $\ep=2^{-N}$, there is $m_N\uparrow\infty$ such that, for all $n\geq m_N$, $P(\ind_A\tilde Z^n_{T^*}>2^{-N})<2^{-N}$. Define
$$A_n:=A\cap\{\ind_A\tilde Z^n_{T^*}\leq 2^{-N}\}\quad\quad\text{for $m_N\leq n< m_{N+1}$}.$$
For $n\geq m_{N_0}$, such that $2^{-N_0}\leq\frac{\alpha}{2}$ we have that
\begin{equation}\label{pff}
P(A_n)\geq P(A)-P(\ind_A\tilde Z^n_{T^*}>2^{-N})\geq \frac{\alpha}2.
\end{equation}

Define the probability measure $\tilde Q^n$ by $\frac{d\tilde Q^n}{dP}:=\tilde Z^n_{T^*}$.
Then, $\tilde Q^n\in \mathbf{M}^{n}_e$ as the density $Z^n_{T^*}$ is a convex combination of densities of equivalent probability measures $Q^k\in \mathbf{M}^{k}_e$, $k\geq n$.
For $m_N\leq n\leq m_{N+1}$, we have that
$$
\tilde Q^n(A_n)=E[\tilde Z^n_{T^*}\ind_A\ind_{\{\ind_A\tilde Z^n_{T^*}\leq 2^{-N}\}}]\leq 2^{-N}.
$$
This shows that $ \tilde Q^n(A_n) \to 0$. As $\tilde Q^n\in\text{conv}(Q^n,Q^{n+1},\dots)$, there is $k_n\geq n$ with $k_n\to\infty$ such that $Q^{k_n}(A_n)\to0$, for $n\to\infty$. As $P\lhd (Q^n)$ it is contiguous with respect to any subsequence of $(Q^{k_n})$ as well, so we should have $P(A_n)\to0$ which is a contradiction to (\ref{pff}). Hence $Z_{T^*}>0$ a.s.
\end{proof}

\begin{proof}[Proof of Theorem~\ref{supermdefllfm}]
Assume that NAA1 holds. Then by Theorem~\ref{aa1result} there exists a sequence of probability measures $Q^n\in\mathbf{M}^n_e$ such that $P\lhd(Q^n)$. Take a strictly positive supermartingale $Z$ which is induced by $(Q^n)$  and $\mathcal{D}=\left([0,T^*]\cap\Q\right)\cup\{T^*\}$ as in Lemma~\ref{komlos_again}.
For all $q\in\mathcal{D}$, denote $Z'_q := \lim_{n \to \infty}\tilde Z^n_q$ where $\tilde Z^n$ are the convex combinations
as in the proof of Lemma \ref{komlos_again}.
Let $X$ be such that $X_{T^*}\in\mathbf{K}^n$ for some $n$. We will show that $(Z_tX_t)_{0 \le t \le T^*}$  is a supermartingale. Indeed, let $r<q$, $r,q\in\mathcal{D}$.
Define,  $\tilde Q^n$ by $\frac{d\tilde Q^n}{dP}:=\tilde Z^n_{T^*}$. As shown in the proof of Lemma \ref{komlos_again},  $\tilde Q^n\in\mathbf{M}^n_e$.
This implies that $(\tilde Z^n_tX_t)_{0 \le t \le T^*}$ is a supermartingale.
Then, by Fatou, we get
\begin{align}
E[Z'_qX_q|\mathcal{F}_r] &= E[\lim_{n\to\infty}\tilde Z^n_qX_q|\mathcal{F}_r]\nonumber\\
&\leq \liminf_{n\to\infty}E[\tilde Z^n_qX_q|\mathcal{F}_r]\nonumber\\
&\leq \lim_{n\to\infty}\tilde Z^n_rX_r =Z'_rX_r.\nonumber
\end{align}
So $(Z'_qX_q)_{q\in\mathcal{D}}$ is a discrete supermartingale. Let now $s<t< T^*$ and $s_k\downarrow s$, $t_j\downarrow t$ for rational $s_k$, $t_j$ (for $t=T^*$ take $t_j\equiv T^*$). Then we have that
\begin{align}
E[Z_tX_t|\mathcal{F}_{s_k}] &= E[\lim_{j\to\infty} Z'_{t_j}X_{t_j}|\mathcal{F}_{s_k}]\nonumber\\
&\leq \liminf_{j\to\infty} E[Z'_{t_j}X_{t_j}|\mathcal{F}_{s_k}]\nonumber\\
&\leq Z'_{s_k}X_{s_k},\nonumber
\end{align}
where the  equality holds by the definition of $Z$ and by the right continuity of $X$, the first inequality is Fatou, the second inequality is the discrete supermartingale property. The right-continuity of the filtration together with the definition of $Z$ and the right continuity of $X$
gives
$$E[Z_tX_t|\mathcal{F}_{s}]=\lim_{k\to\infty}E[Z_tX_t|\mathcal{F}_{s_k}]\leq\lim_{k\to\infty}Z'_{s_k}X_{s_k}=Z_sX_s.$$
Hence, $ZX$ is a supermartingale.

For the converse, assume  that there is a supermartingale deflator for the large financial market. Suppose there exists an asymptotic arbitrage of first kind, that is, there exists a sequence $X^k_{T^*}\in\mathbf{K}^{n_k}$ such that $X^k_t\geq-\ep_k$, $0\leq t\leq T^*$, and $P(X^k_{T^*}\geq C_k)\geq\alpha$ with $\ep_k\to0$ and $C_k\to\infty$.
We have that $X^kZ$ are supermartingales, for all $k$. Hence, as $X^k_0=0$,
\begin{equation}\label{leq0}
E[X^k_TZ_T]\leq X^k_0Z_0=0.
\end{equation}
On the other hand, let $A_k:=\{X^k_{T^*}\geq C_k\}$. As $Z_0\leq1$ and by the properties of $X^k$,
\begin{equation}\label{>0}
E[X^k_{T^*}Z_{T^*}]\geq C_kE[Z_{T^*}\ind_{A_k}]-\ep_k E[Z_{T^*}]\geq C_kE[Z_{T^*}\ind_{A_k}]-\ep_k.
\end{equation}
By assumption  $P(A_k)\geq\alpha$, for all $k$.
We claim that there exists $\beta>0$ such that,  $P(\{Z_{T^*}>\beta\}\cap A_k)\geq\frac{\alpha}{2}$ for all $k$.  Suppose not, then
for each $j\geq 1$ and $\beta=\frac1{j}$, there is $k_j$ such that $P(\{Z_{T^*}>\frac1{j}\}\cap A_{k_j})<\frac{\alpha}{2}$ and hence
$P(\{Z_{T^*}\leq\frac1{j}\}\cap A_{k_j})\geq P(A_{k_j})-\frac{\alpha}{2}\geq\frac{\alpha}2$. Therefore
$$P(Z_{T^*}=0)=\lim_{j\to\infty}P(Z_{T^*}\leq\frac1{j})\geq \liminf_{j\to\infty}P(\{Z_{T^*}\leq\frac1{j}\}\cap A_{k_j})\geq\frac{\alpha}{2},$$
a contradiction to the integrability of $Z_{T^*} $ and $Z_{T^*}>0$ a.s. Equation (\ref{>0}) then implies that
$$E[X^k_{T^*}Z_{T^*}]\geq C_kE[Z_{T^*}\ind_{A_k\cap\{Z_{T^*}>\beta\}}]-\ep_k\geq C_k\beta\frac{\alpha}2-\ep_k,$$
which is strictly positive for $k$ large enough. This gives a contradiction to (\ref{leq0}).

We still have to prove (1) and (2). Assume that $S^i_t\geq -a$, $0\leq t\leq T^*$ for some $i$.
Define the trivial predictable $H_t=\ind_{]0,T^*]}(t)$, then, for $t\leq T^*$,
$$X_t=(H\cdot S^i)_t=S^i_t-S^i_0\geq -a-S^i_0.$$
Hence $X_{T^*}\in \mathbf{K} ^n$, for $n\geq i$. Therefore $(X_tZ_t)_{0 \le t \le T^*}$ is a supermartingale. This immediately gives that $S^i_tZ_t=X_tZ_t + S^i_0Z_t$ is a supermartingale as $S^i_0\geq0$, which shows  (1).

Denote $S^i=S$. For (2) note that we can stop $S$ such that, for all $n$, there is a $\kappa_n\geq 0$ and $|S^{\tau_n}_t|\leq\kappa_n$.
Then for $H^+_t=\ind_{\lsto 0,\tau_n\rsto}(t)$ and $H^-_t=-\ind_{\lsto 0,\tau_n\rsto}(t)$ we have that
$X_t=(H^+\cdot S)_t=S_{\tau_n\wedge t}-S_0\geq -\kappa_n-S_0$ and $-X_t=(H^-\cdot S)_t=-S_{\tau_n\wedge t}+S_0\geq -\kappa_n+S_0$.
Therefore $X_{T^*}$ and $-X_{T^*}$ are in $\mathbf{K}^n$, for $n\geq i$. Hence, for all $n$, $\pm Z_t(S_{\tau_n\wedge t}-S_0)$, $0\leq t\leq T^*$, are supermartingales. Now, $S_0=0$, so
$\pm Z_tS_{\tau_n\wedge t}$, $0\leq t\leq T^*$, are supermartingales. Stopped supermartingales are supermartingales, hence $\pm(ZS)^{\tau_n}$ are supermartingales, and therefore   $(ZS)^{\tau_n}$ is a martingale. This holds for all $n$, and so $(Z_tS_t)_{0\leq t\leq T^*}$ is a local martingale.
\end{proof}

In the sequel we apply the results to bond markets: again, for each $ T\in[0,T^*]$, $(P(t, T))_{0\leq t\leq T}$ is a strictly positive c\`adl\`ag stochastic process adapted to
$(\mathcal{F}_t)_{0\leq t\leq  T}$ with $P( T, T)=1$. We assume that, for fixed $t$, the function $ T\mapsto P(t, T)$ is almost surely right-continuous. Note, that in this section, we do not have any local boundedness assumptions on $P(t, T)$ or $\frac1{P(t,T)}$. We will again have to assume that in the case of a finite number of assets discounted with the num\'eraire $P(t,T^*)$ we will not have any arbitrage opportunities. This is again a consequence of the following assumption on existence of the $T^*$-forward measure. Consider any $0<T_1<\dots<T_n \leq T^*$ and define the cone $\mathbf{C}(T_1,\dots,T_n,T^*)$  as in (\ref{K}) where $S^i$ is defined as in  \eqref{defSi}.
Note that we do not assume here that $T^*\in\cT$.

\begin{assumption}\label{sigma}
For all finite collections of maturities $T_1<T_2<\dots<T_n \leq T^*$  there
exists a measure $Q \sim P|_{\mathcal{F}_{T^*}}$ such that for all $f\in \mathbf{C}(T_1,\dots,T_n,T^*)$ we have that $E_{Q}[f]\leq 0$.
\end{assumption}
The condition $E_{Q}[f]\leq 0$ means that the measure  $Q \in\mathbf{M}_e(T_1,\dots,T_n,T^*)$, where the definition of the set of separating measures is analogous as in (\ref{Me}).
Note that Assumption~\ref{sigma} implies the existence of an equivalent sigma-martingale measure for $S^1,\dots,S^n$ given as as in \eqref{defSi}, and therefore these processes are semimartingales, see \cite{Delb:Schach:1998}.
As in Definition~\ref{LFM}, any sequence $(T_i)_{i\in\N}$ in $[0,T^*]$ induces a large financial market.

\begin{definition}\label{naa1}
The bond market satisfies NAA1 w.r.t.~a sequence $(T_i)_{i\in\N}$ in $[0,T^*]$ if for the large financial market induced
by $(T_i)_{i\in\N}$ there does not exist an asymptotic arbitrage of first kind.
\end{definition}

\begin{theorem}\label{supermdefl}Fix a sequence $(T_i)_{i\in\N}$ in $[0,T^*]$.
The bond market satisfies NAA1 w.r.t.~$(T_i)_{i\in\N}$ if and only if there exists a strictly positive supermartingale deflator $(Z_t)_{0\leq t\leq T^*}$ for the large financial market induced by  $(T_i)_{i \in \N}$.
If $(T_i)_{i\in\N}$ is dense in $[0,T^*]$, then $\left(Z_t\frac{P(t, T)}{P(t,T^*)}\right)_{0\leq t\leq T}$ is a supermartingale for all $ T\leq T^*$.
\end{theorem}

\begin{proof}[Proof of Theorem~\ref{supermdefl}]
Everything follows by Theorem~\ref{supermdefllfm}. It only remains
   to show, that  $\frac{P(t,T)}{P(t,T^*)}Z_t$ is a supermartingale for each $ T\leq T^*$ which is not an element of the dense sequence. Note that for $T^*$ the statement holds, as $\frac{P(t,T^*)}{P(t,T^*)}Z_t=Z_t$ is a supermartingale. Let $T<T^*$.
Choose $\tilde T_i\downarrow T$ with $\tilde T_i$ elements of the dense sequence in $[0,T^*]$. Let $X_t:=\frac{P(t,T)}{P(t,T^*)}$ and $X^i_t:=\frac{P(t,\tilde T_i)}{P(t,T^*)}$. As, for each $i$,  $X^i_t>0$ a.s., for all $t$, we get by Theorem~\ref{supermdefllfm}, (1), that $Z_tX^i_t$ is a supermartingale. Hence
\begin{align}
E[X_tZ_t|\mathcal{F}_s] &=E[\lim_{i\to\infty}X^i_tZ_t|\mathcal{F}_s]\nonumber
&\leq \liminf_{i\to\infty}E[X^i_tZ_t|\mathcal{F}_s]\nonumber
&\leq
\lim_{i\to\infty}X^i_sZ_s=X_sZ_s,\nonumber
\end{align}
as $U\mapsto P(v,U)$ is right-continuous, for each v (and therefore $X^i_v\to X_v$ for $v=s,t$) and by Fatou.

\end{proof}

Finally we will show that under the weaker assumptions of this section we will still be able to define a generalized bank account.

\begin{theorem}\label{super_defl_ba}
Let $(T_i)$ be a dense sequence in $[0,T^*]$ such that NAA1 holds.
Let $(B^n_t)_{t\in[0,T^*]}$ be the sequence of roll-overs as in Definition~\ref{rollover}, where the refining partition $\{t^n_1,\dots, t^n_{k_n}\}$ is chosen such
$\bigcup_{n\in\N}\{t^n_1,\dots, t^n_{k_n}\}\subseteq(T_i)$. Then
there exists a sequence of convex combinations $\tilde{B}^n\in\text{conv}(B^n,B^{n+1},\dots)$ and a c\`adl\`ag stochastic process $(B_t)_{0\leq t\leq T^*}$ (the generalized bank account) such that
$$B_t=\lim_{q\downarrow t}\lim_{n\to\infty}\tilde{B}^n_q,$$
with $B_0\leq 1$ and $0\leq B_t<\infty$, for all $t\leq T^*$. The generalized bank account has the following properties.
\begin{enumerate}
\item The process $(V_t)_{0\leq t\leq T^*}$, where $V_t=Z_t\big(\frac{B_t}{P(t,T^*)}-\frac1{P(0,T^*)}\big)$, is a supermartingale, where $Z$ is the supermartingale deflator as in Theorem~\ref{supermdefl}. This implies that the bank account $B$ discounted with respect to $P(.,T^*)$ multiplied by $Z$ is a supermartingale as
    $$\frac{Z_t B_t}{P(t,T^*)}=V_t+\frac{Z_t}{P(0,T^*)}.$$
\item If $0<P(t,T)\leq1$, for all $ T\leq T^*$, then $P(B_t\ge1)=1$, for all $t\leq T^*$.
\end{enumerate}
\end{theorem}

The interpretation of this Theorem is, that for any refining sequence of partitions, which does not produce an asymptotic arbitrage opportunity of first kind in the induced large financial market, there does exist a generalized bank account.
In particular, if the bond market does not allow an asymptotic arbitrage opportunity of first kind for any sequence of maturities in $[0,T^*]$ (for the respective induced large financial market as in Definition~\ref{LFM}), then any refining sequence of partitions gives a generalized bank account in the sense of Theorem~\ref{super_defl_ba}. If, moreover, the bond $P(t, T)\leq 1$, $0\leq t\leq T\leq T^*$, then we can say that $B_t$ is bounded from below by $1$ a.s. This corresponds to the case, where a non-negative short rate exists.

\begin{proof}
By Lemma~\ref{selff} we have that $\frac{B^n_{T^*}}{P(t,T^*)} - \frac{1}{P(0,T^*)} \in\mathbf{K}^{m_n}$ for some $m_n$ large enough. By NAA1 and Theorem~\ref{supermdefl} there exists a strictly positive c\`adl\`ag supermartingale $Z$ such that, for all $n$, $(V^n_t)_{t\in[0,T^*]}$ is a supermartingale, where $V^n_t:=Z_t\big(\frac{B^n}{P(t,T^*)}-\frac1{P(0,T^*)}\big)$, since all points of the partition defining the roll-over bond are contained in the dense sequence instrumental for the definition of $Z$. As in the proof of Theorem~\ref{bank_account} we get a sequence of convex combination $\tilde V^n_t\in\text{conv}(V^n,V^{n+1},\dots)$ and a c\`adl\`ag supermartingale $0\leq V_t<\infty$ such that $V_t=\lim_{q\downarrow t}\lim_{n\to\infty}\tilde V^n_q$. Moreover,
$V_0\leq 0$ as $V_0\leq\lim_{n\to\infty}\tilde V^n_0=0$.
Define $$B_t=\big(\frac{V_t}{Z_t}+\frac1{P(0,T^*)}\big)P(t,T^*).$$
 Similarly as in the proof of Theorem~\ref{bank_account}, we see that $B_t=\lim_{q\downarrow t}\lim_{n\to\infty}\tilde B^n_q$, where $\tilde B^n_q$ are the corresponding convex combinations of $B^n_q$, i.e.,
 $\tilde B^n_q=\big(\frac{\tilde V^n_q}{Z_q}+\frac1{P(0,T^*)}\big)P(q,T^*)$.
(Use the right continuity of $t\to P(t,T)$ and $t\to Z_t$.) Clearly $B_0=\big(\frac{V_0}{Z_0}+\frac1{P(0,T^*)}\big)P(0,T^*)\leq1$.
\end{proof}


\section{examples}\label{section7}

\subsection{A strict local martingale deflator}\label{benchmark-example}
In this section we consider the example touched upon in the introduction in more detail.
Let $S^*$ denote the growth optimal portfolio. In the benchmark approach fair prices of
a payoff $X$ at time $T>0$ are given by
$$
\pi_t(X) = 	S_t^* E_Q[\frac{X}{S_T^*}| \cF_t], $$
and, as a consequence, one obtains bond prices of the form
\begin{equation}\label{def:bond}
P(t,T)=E_Q \big[ \frac{S_t^*}{S^*_T}| \mathcal{F}_t \big].
\end{equation}
We give an example where the inverse of the growth optimal portfolio is related to a strict local
martingale. Consider the case where
$$ \frac{1}{S_t^*} = \frac{A(t)}{\parallel x + W_t \parallel^2}=:\xi_t $$
with a positive, deterministic, c\`adl\`ag function $A:[0,\infty) \mapsto (0,\infty)$, a four-dimen\-sional standard Brownian motion $W$ and $0 \not = x \in \R^4$.
Then $(\parallel x + W_t \parallel^2)_{t \ge 0}$ is a squared Bessel process of dimension four and
its inverse is a strict local martingale.

\noindent In this example, for each $T^*$,  there exists an equivalent probability measure $Q^*$ such that all bond prices with maturity $T\leq T^*$ discounted by the num\'eraire $P(.,T^*)$ are martingales under $Q^*$. So we are in the situation of Theorem~\ref{th1}. Indeed,  let $\alpha=\frac1{E_Q[\xi_{T^*}]}$ and define
$$\frac{dQ^*}{dQ}=\alpha\xi_{T^*}.$$
 The density process $Z^*_t$, $0\leq t\leq T^*$, satisfies
$$Z_t^*=\alpha E_{Q}[\xi_{T^*}|\mathcal{F}_t]=\alpha\xi_tP(t,T^*).$$
Therefore
$$Z_t^*\frac{P(t,T)}{P(t,T^*)}=\alpha E_{Q}[\xi_T|\mathcal{F}_t],$$
hence $\frac{P(t,T)}{P(t,T^*)}$, $0\leq t\leq T$, is a martingale with respect to $Q^*$.

Using Markovianity of $W$ and integrating over the transition density of  squared Bessel processes one obtains the following explicit expression for $P(t,T)$.
\begin{align} \label{tslocal}
P(t,T)
&= \frac{A(T)}{A(t)} E_Q\left[\frac{\parallel x + W_t \parallel^2}{\parallel x + W_T \parallel^2}\,\Big| \,\cF_t\right]
=\frac{A(T)}{A(t)} \left(1-e^{- \frac{\parallel x+W_t\parallel^{-2}}{2(T-t)}}\right),
\end{align}
see Equation 8.7.17 in \cite{HeathPlaten}.

Assume that $B^n_t$ is defined as in Definition~\ref{rollover}, where we additionally assume that there exists a constant $K\geq1$ such that, for all $n$, $\max_{1\leq i\leq k_n}|t^n_i-t^n_{i-1}|\leq \frac{K T^*}{k_n}$. Then we get the following.
\begin{lemma}\label{limit}
$B^n_t \to B_t:=\frac{A(0)}{A(t)}$ a.s., for $n\to\infty$.
\end{lemma}
\begin{proof}
Let $0=t^n_0<t^n_1\dots<t^n_{k_n}=T^*$ and
$\delta_n:=\max_{1\leq i\leq {k_n}}(t^n_{i}-t^n_{i-1})$. By assumption $\delta_n\leq \frac{K T^*}{k_n}\to 0$ for $n\to\infty$.  Fix $t$, then for each $n$ there is $j_n\leq k_n$ such that $t^n_{j_n-1}< t\leq t^n_{j_n}$. We have that
$$B^n_t =\prod_{i=1}^{j_n}\frac{1}{P(t^n_{i-1},t^n_i)}P(t,t^n_{j_n})
=\frac{A(0)}{A(t^n_{j_n})}\prod_{i=1}^{j_n}\big(1-e^{-\frac{\|x+W_{t^n_{i-1}}\|^{-2}}{2(t^n_i-t^n_{i-1})}}\big)^{-1}P(t,t^n_{j_n}).
$$
By the right continuity of $A$ we have that $A(t^n_{j_n})\to A(t)$ for $n\to\infty$ and it is clear that $P(t,t^n_{j_n})\to1$ for $n\to\infty$. Moreover, as $0<1-e^{-\frac{\|x+W_{t^n_{i-1}}\|^{-2}}{2(t^n_i-t^n_{i-1})}}\leq 1$, for all $i$, we get
$$\prod_{i=1}^{j_n}\big(1-e^{-\frac{\|x+W_{t^n_{i-1}}\|^{-2}}{2(t^n_i-t^n_{i-1})}}\big)^{-1}\geq1.$$
We will now show that
$$\lim_{n\to\infty}\prod_{i=1}^{j_n}\big(1-e^{-\frac{\|x+W_{t^n_{i-1}}\|^{-2}}{2(t^n_i-t^n_{i-1})}}\big)^{-1}\leq1 \quad \text{a.s.},
$$
which then implies that $B^n_t \to B_t=\frac{A(0)}{A(t)}$. Let $m_t:=\min_{s\leq t}\|x+W_s\|^{-2}$ be the running minimum of the inverse of the squared Bessel process of dimension 4. We have that $m_t>0$ a.s. By assumption $\delta_n\leq \frac{K T^*}{k_n}$. Hence
\begin{align}
\prod_{i=1}^{j_n}\big(1-e^{-\frac{\|x+W_{t^n_{i-1}}\|^{-2}}{2(t^n_i-t^n_{i-1})}}\big)^{-1}
&\leq\prod_{i=1}^{j_n}\big(1-e^{-\frac{k_nm_t}{2K T^*}}\big)^{-1}\nonumber\\
=\big(1-e^{-\frac{k_nm_t}{2K T^*}}\big)^{-j_n}&\leq \big(1-e^{-\frac{k_nm_t}{2K T^*}}\big)^{-k_n},\nonumber
\end{align}
as, clearly, $0<1-e^{-\frac{k_nm_t}{2K T^*}}\leq 1$. But as for almost all $\omega$ we have that $m_t(\omega)>0$ and as
$(1-e^{-ak_n})^{-k_n}\to1$ for $n\to\infty$ and $a>0$ we get that
$$\lim_{n\to\infty}\big(1-e^{-\frac{k_nm_t}{2K T^*}}\big)^{-k_n}=1\text{ a.s.}
$$
\end{proof}

By Theorem~\ref{bank_account},  $(B_t)_{0\leq t\leq T^*}$ discounted with respect to the num\'eraire $P(.,T^*)$ is a $Q^*$--supermartingale. Lemma~\ref{limit} and the definition of the measure $Q^*$ moreover gives  that $\big(\frac{B_t}{P(t,T^*)}\big)_{0\leq t\leq T^*}$  is a strict local martingale under $Q^*$ as

$$Z^*_t\frac{B_t}{P(t,T^*)}=\alpha\xi_tB_t=\alpha\frac{A(0)}{\|x+W_t\|^2};$$
hence it does not qualify as strong nor as weak num\'eraire.

Besides the term structure given by $(P(t,T))_{0 \le t \le T}$,  we observe a second virtual term structure in correspondance with Definition \ref{def:virtual}: if there was enough liquidity such that the bank account qualifies as num\'eraire,
one can price with the supermartingale deflator
$$
\frac{1}{E_{Q^*}[\frac{B_T}{P(T,T^*)}]} \frac{B_T}{P(T,T^*)}  \frac{1}{B_T},
$$
which stems from changing measure by the supermartingale $ \frac{B_T}{P(T,T^*)} $.
Pricing $ 1 $ at time $ T $ with respect to this deflator, we obtain
the  virtual term structure
\begin{align}
\tilde{P}(0,T) &= B_0 E_{Q^*} \Big[  \frac{B_T}{P(T,T^*)} \frac{1}{E_{Q^*}[\frac{B_T}{P(T,T^*)}]} \frac{1}{B_T} \Big] \nonumber \\
&= \frac{B_0}{B_T} = \frac{A(T)}{A(0)}. \label{tsvirtual}
\end{align}
The illiquidity premium for the claim $Y\equiv 1$ at time $0$ can be computed from \eqref{eq:liquidity} and we obtain
\begin{align*}
 1- \big(E_{Q^*}\big[\frac{B_T}{P(T,T^*)}\big]\big)^{-1} &= 1-\frac{1}{ A(0) \, E_{Q^*}\big[(A(T)P(T,T^*))^{-1}\big]} \\
 &= 1-\frac{A(T^*)}{ A(0) \, E_{Q^*}\Big[\Big( 1-e^{- \frac{\parallel x+W_T\parallel^{-2}}{2(T^*-T)}}   \Big)^{-1}\Big] } ,
\end{align*}
by \eqref{tslocal}. The expectation is given in terms of the transition density of Bessel processes.

In Figure \ref{fig:localmartingaledeflator} we consider the case where $A_t \equiv 1$ and show the term structure given by \eqref{tslocal} in comparison to the virtual  term structure given in \eqref{tsvirtual}.  In Section \ref{sec:exotic} we will meet an example where investing in the roll-over account may even lead to a total loss of invested money.

\begin{figure}[h]
\begin{overpic}[width=8cm]{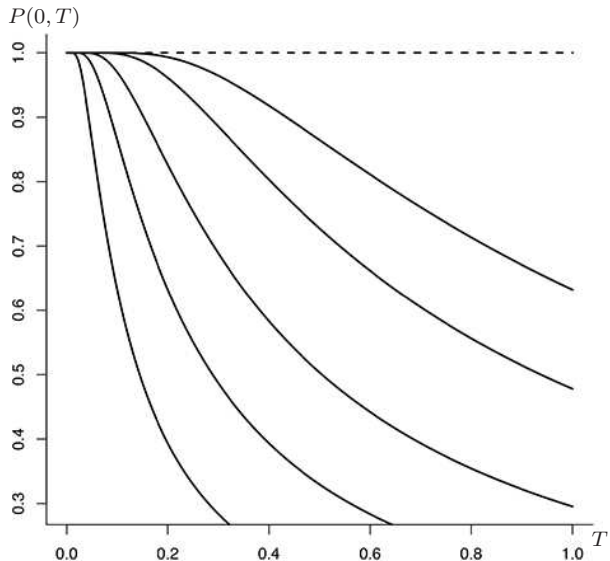}
\put(98,4){{\footnotesize $T$}}
\put(1,91){{\footnotesize $P(0,T)$}}
\end{overpic}
\caption{\label{fig:localmartingaledeflator} This figure illustrates the two term structures: the term structure from Equation \eqref{tslocal} with $A_t\equiv 1$, is shown by the lines $T\mapsto P(0,T)$ for different $\parallel x \parallel^{-2} \in\{0.2,0.4,0.7,1.3,2\}$. The constant virtual term structure $T \mapsto \tilde P(0,T) \equiv 1$ from Equation \eqref{tsvirtual} is represented by the dashed line.}
\end{figure}

\subsection{Bond markets driven by fractional Brownian motion}
The purpose of this section is to illustrate the applicability of our approach beyond semimartingale models. The semimartingale assumption is standard in the literature on bond markets (see  for example \cite{DoeberleinSchweizer2001} and the referenced literature therein), while we were able to show that this is in general not necessary.
In this regard, we study some models for  bond markets driven by fractional Brownian motion. In this models, the bank account or the bond prices may no longer be semimartingales in this setting, while \emph{discounted} prices are, such that an appropriate no-arbitrage condition still holds.

More precisely, we first consider  a locally integrable short rate process which is given by a time-inhomogeneous variant of fractional Brownian motion and show that discounted bond prices turn out to be martingales. While in this example the bank account is absolutely continuous and hence a semimartingale, we show in Remark \ref{rem:FB} that it is also possible to consider a bank account directly driven by the fractional Brownian motion. In this setting, no short rate exists and  the bank account is no longer a semimartingale. 

Regarding \eqref{bondprice} we need to obtain the conditional distribution of a fractional
Brownian motion, which we establish following \cite{PipirasTaqqu:2001};  see also \cite{FinkZaehleKlueppelberg2012} for related results.
A fractional Brownian motion (FBM) with Hurst parameter $H\in(0,1)$ is a zero-mean stationary Gaussian process $Z=Z^H$ with covariance function
\begin{align*}
E[Z_s Z_t] &= \frac{1}{2} \big( |s|^{2H} + |t|^{2H} - |s-t|^{2H} \big).
\end{align*}
For $H=\frac{1}{2}$, $Z$ is a standard Brownian motion. If $H>\frac{1}{2}$ the fractional Brownian motion has long-range dependence. Moreover, for $H\not =\frac{1}{2}$, $Z$ is no longer a semimartingale.

To ease the exposition we consider $H>\frac{1}{2}$ only. Define the right-sided fractional Rieman-Liouville integral of order $\alpha >0$ by 
\begin{align*}
\big( I_{t-}^\alpha f\big) (s) &:= \frac{1}{\Gamma(\alpha)} \int_s^t f(u) (u-s)^{\alpha-1} du, \quad s \in (0,t).
\end{align*}
$I^0$ is the identity. The fractional derivative of order $0<\alpha<1$ is denoted by $I^{-\alpha}$, i.e.
\begin{align*}
\big( I_{t-}^{-\alpha} f\big) (s) &:= - \frac{1}{\Gamma(1-\alpha)} \frac{d}{ds} \int_s^t f(u) (u-s)^{-\alpha} du,
\quad s \in (0,t).
\end{align*}
For the further analysis it will be useful to consider $\kappa := H-\frac{1}{2}$ instead of $H$ itself.
Fix a finite time horizon $T^*>0$.
Let\footnote{
We write $\cdot^{-\kappa} f(\cdot)$  short for the function $u \mapsto u^{-\kappa} f(s)$.}
\begin{align*}
(K_\kappa \, f)(s) := c_\kappa s^{- \kappa} \, \Big( I_{T^*-}^{\kappa} (\cdot^{\kappa} f(\cdot) ) \Big) (s),
\end{align*}
with constant $c_\kappa = \sqrt{ \frac{\pi \kappa (2 \kappa +1) }{\Gamma(1-2\kappa) \sin(\pi \kappa) } }$.
The adjoint operator of $K_\kappa$ is
\begin{align*}
(K^*_\kappa \, f)(s) := c_\kappa s^{- \kappa} \, \Big( I_{T^*-}^{-\kappa} (\cdot^{\kappa} f(\cdot)) \Big) (s).
\end{align*}
It turns out that for $H>\frac{1}{2}$ the proper space of deterministic integrands to consider is, see \cite{PipirasTaqqu:2001},
$$ \Lambda_{T^*}^\kappa := \big\{ f: \exists \, \phi_f \in L^2[0,T^*], \ \text{s.t. } f(s) = (K_\kappa^* \phi_f) (s) \big\}.
$$
Then $ \Lambda_{T^*}^\kappa$ is a Hilbert space with corresponding norm
$$ \parallel f \parallel _{\Lambda_{T^*}^\kappa} := \parallel K_\kappa f \parallel_{L^2([0,T^*])}.$$
The integral of $f\in\Lambda_{T^*}^\kappa$ w.r.t.~the fractional Brownian motion $Z$ is obtained as the limit of $\int f_n dZ$ with elementary $f_n$ s.t.\
$\parallel f_n - f \parallel_{\Lambda_{T^*}^\kappa} \to 0$. Of course, for elementary $f$, say $f= \sum a_i \ind_{(s_i,t_i) }$
the integral equals $\int f dZ = \sum a_i (Z_{t_i}-Z_{s_i})$.
Letting $k^\kappa(t,s):=(K^\kappa \ind_{[0,t]})(s)$  the covariance function of $Z$   has the following representation
\begin{align}\label{eq:Rkappa}
R_\kappa (t,s) := E[Z_s Z_t] = \int_0^t k^\kappa(t,w) k^\kappa(s,w) dw.
\end{align}
For $\kappa = 0$ we obtain $K_\kappa = \text{id}$, i.e.\ $R_\kappa(t,s) = s \wedge t$
which is the covariance function of a Brownian motion.
From the representation in \eqref{eq:Rkappa} it is immediate that
\begin{align}\label{eq:reprFBM}
Z_t \stackrel{\ccL}{=} \int_0^t k^\kappa(t,w) dB_w
\end{align}
where $B$ is a standard Brownian motion. This result was already discovered in the seminal work of \cite{Mandelbrot:1968} and leads
to the following representation of conditional expectations (Theorem 7.1 in \cite{PipirasTaqqu:2001}): let $\cF^Z_t := \sigma(Z_s: 0 \le s \le t)$.
For $0<s<t$
\begin{align}\label{eq:condexpectationFBM}
E \big[ Z_u | \cF^Z_t \big] = Z_t + \int_0^t \psi_u(t,w) dZ_w
\end{align}
with
$$ \psi_u(t,w) = \psi^\kappa_u(t,w) :=  \frac{\sin (\pi \kappa)}{\pi} w^{-\kappa} (t-w)^{- \kappa} \int_t^u \frac{z^\kappa (z-t)^\kappa}{z-w} dz. $$
Note that
$$ \psi_u(t,w) = w^{-\kappa} (I_{t-}^{-\kappa} ( I_{u-}^\kappa (\cdot^\kappa \ind_{[t,u)}))) (w). $$
Proceeding similarly, we are able to compute the conditional covariance of $Z$.
\begin{lemma} For $0 < t < u, v$, \label{lem:condVarianceFBM}
\begin{align*}
E[Z_u Z_v|\cF^Z_t]  = E[Z_u|\cF_t^Z]\cdot E[Z_v|\cF_t^Z] + \int_t^{ u \wedge v} k^\kappa(u,w) k^\kappa(v,w) dw.
\end{align*}
\end{lemma}
\begin{proof}
The proof mainly relies on \eqref{eq:reprFBM}. We have that 
\begin{align}
\lefteqn{E\Big[ \int_0^u k^\kappa (u,w) dB_w \int_0^v k^\kappa (v,w) dB_w | \cF^B_t \Big] } \hspace{2cm} \nonumber\\
& = \int_0^t  k^\kappa(u,w) dB_w \, \int_0^t k^\kappa (v,w) dB_w  \label{temp237}\\
&+ E\Big[ \int_t^u k^\kappa(u,w) dB_w \int_t^v k^\kappa(v,w) dB_w |\cF^B_t\Big]. \nonumber
\end{align}
As standard Brownian motion has independent increments, the last expectation is easily computed, leading to the last term in our result.
It remains to represent the first addend in terms of $Z$. Using \eqref{eq:reprFBM} we obtain 
\begin{align*}
E[Z_u|\cF_t^Z] &= \int_0^t k^\kappa(u,w) dB_w
\end{align*}
and we conclude. \end{proof}

With this results at hand we are ready to consider bond markets where the short rate is driven by  fractional Brownian motion.
Fix a measure $Q\sim P$ and assume that $Z$ is a FBM with parameter $\kappa$ under $Q$.  As num\'eraire we consider the bank account $B(t)=\exp(\int_0^t r_u du)$   where the short rate is given by
\begin{align} r_t = \mu(t) + \sigma(t) Z_t, \quad 0 \le t \le T^*,\label{eq:rFBM} \end{align}
with $\mu: [0,T^*] \mapsto \R^+$  in $L^1[0,T^*]$ and  $\sigma: [0,T^*] \mapsto \R^+$ being
an element of $\Lambda_{T^*}^\kappa.$ Bond prices are given by
\begin{align} \label{fbm:bonds}
 P(t,T) = E_Q[ \frac{B_t}{B_T} | \cF_t], \qquad 0 \le t \le T \le T^*. 
\end{align}

\begin{proposition}\label{bondfractional}
Let $\mu^*(t,T) := \int_t^T \mu(u) du$ and $\sigma^*(t,T) := \int_t^T \sigma(u) du$.
Under \eqref{eq:rFBM} the bond prices equal
\begin{align*}
P(t,T) &= \exp\bigg[ \mu^*(t,T) + \sigma^*(t,T) Z_t + \int_t^T \int_0^t \sigma(u)   \psi_u(t,w) dZ_w \, du \nonumber\\
&+\frac{1}{2} \int_t^T \int_t^T \int_0^{u \wedge v} \sigma(u) \sigma(v) k^\kappa(u,w) k^\kappa(v,w) \,dw \,du \,dv \bigg], \quad 0 \le t \le T \le T^*.
\end{align*}
\end{proposition}
\begin{proof}
First, note that $J(t,T) := \int_t^T r_u du$ is a Gaussian process with
\begin{align*}
E_Q[J(t,T)|\cF_t] & = \int_t^T \Big( \mu(u) + \sigma(u) E_Q[Z_u |\cF_t] \Big) du \\
&= \mu^*(t,T) + \int_t^T  \sigma(u) (Z_t + \int_0^t \psi_u(t,w) dZ_w)  du \\
&= \mu^*(t,T) + \sigma^*(t,T) Z_t +  \int_t^T  \sigma(u)  \int_0^t \psi_u(t,w) dZ_w\,   du ,
\end{align*}
using \eqref{eq:condexpectationFBM}. For the conditional variance of $J(t,T)$ note that
\begin{align}\label{temp:variance}
\Var[J(t,T) | \cF_t] &= E_Q \bigg[\bigg(\int_t^T \sigma(u) \Big(Z_u - E_Q[Z_u|\cF_t]\Big) du\bigg)^2 \big|\cF_t\bigg] .
\end{align}
By Lemma  \ref{lem:condVarianceFBM} we obtain
\begin{align*}
E_Q[Z_u Z_v |\cF_t] - E_Q[Z_u|\cF_t] E_Q[Z_v|\cF_t] =
\int_t^{ u \wedge v} k^\kappa(u,w) k^\kappa(v,w) dw.
\end{align*}
Inserting this into \eqref{temp:variance} gives that
\begin{align*}
\Var[J(t,T) | \cF_t] &= \int_t^T \int_t^T \sigma(u) \sigma(v) \int_t^{ u \wedge v} k^\kappa(u,w) k^\kappa(v,w) dw\, dv \, du.
\end{align*}
Finally, note that $J(t,T)$ is, conditional on $\cF_t$, a Gaussian random variable. Using its Laplace-Stieltjes transform  we conclude.
\end{proof}
By \eqref{fbm:bonds}, discounted bond prices are $Q$-martingales. Moreover, from  Lemma \ref{lem:2.4} it follows that  Assumption \eqref{numemm} holds. Hence $Q$ is an ELMM such that by Theorem \ref{th1} NAFL holds.

\begin{remark}\label{rem:FB}
In \cite{DoeberleinSchweizer2001} semimartingale models are covered while we drop this assumption
in our setup. Consider for example the case where
$ B_t = \exp ( Z_t )$ such that for $H \not = \frac{1}{2}$ the bank account is not a semimartingale.
Analogously to Proposition \ref{bondfractional}, bond prices can be computed and are
$Q$-martingales. Again, by Lemma \ref{lem:2.4} this case is included in our setup.
\end{remark} 

\subsection{An extension of the HJM setup}

In this section we extend the HJM setup by an additional component which is not absolutely continuous in terms of maturity, such that, in general, a short rate  does not exist in this framework. Using Theorem \ref{th1} we classify those models which satisfy NAFL by means of a generalised drift condition. The HJM-model is contained as special case.

Fix a finite time horizon $T^*$ and a measure $Q\sim P |_{T^*}$. There are two independent $Q$-Brownian motions $W$ and $V$ where $W$ is $d$-dimensional and $V$ is one-dimensional.  We consider the filtration $(\cF_t)_{t \ge 0}$ given by
$$ \cF_t = \sigma( W_s: 0 \le s \le t, V_u : u \ge 0) \vee N $$
which is the initial enlargement of the natural filtration of $W$ with the full path of $V$ and all $P$-nullsets $N$.
We assume that  bond prices are given by
\begin{align}\label{HJMsetup}
P(t,T) = \exp\Big( - \int_t^T f(t,u) d V(u) - \int_t^T g(t,u) du \Big), \quad 0 \le t \le T \le T^*
\end{align}
with families of It\^o-processes $f$ and $g$ to be specified below.
This includes the HJM-framework if $f\equiv 0$. In the following we characterize when the considered measure $Q$ is an (equivalent) local martingale measure in the sense used in Theorem \ref{th1}. All models which satisfy NAFL are given by an equivalent change to such a local martingale measure.

For given initial curves $T \mapsto f(0,T)$ and $T \mapsto g(0,T)$ we assume that $f$ and $g$ satisfy
\begin{align}
f(t,T) &= f(0,T) + \int_0^t a(s,T) ds + \int_0^t b(s,T) dW_s, \label{HJM:f}\\
g(t,T) &= g(0,T) + \int_0^t c(s,T) ds + \int_0^t d(s,T) dW_s, \label{HJM:g}
\end{align}
for $0 \le t \le T \le T^*$. Denote by $\mathcal{O}$ the optional sigma-algebra on $\Omega\times \R_+$. We assume the following regularity conditions:
\begin{align}
\tag{HJM1} &a,b,c \text{ and } d \text{ are }\mathcal{O} \otimes \mathcal{B}(\R_+)-\text{measurable}, \\
\tag{HJM2} &\int_0^{T^*} \int_0^{T^*} (|a(s,t)| + |c(s,t)| )ds \, dt < \infty, \\
\tag{HJM3} &\sup_{0 \le s \le t \le T^*}(\parallel b(s,t) \parallel + \parallel d(s,t) \parallel) < \infty.
\end{align}
Recall that $Q$ is an equivalent local martingale measure (ELMM) if
$$ \left(\frac{P(t,T)}{P(t,T^*)}\right)_{0\leq t\leq T}
\text{ is a local martingale for all } T\in[0,T^*]. $$
For any $T \le T^*$ we set
$$ A(t,T) := \int_T^{T^*} a(t,u) dV_u, \quad C(t,T) := \int_T^{T^*} c(t,u) du,$$
and similar for $b$ (as on the left hand side) and $d$ (as on the right hand side).

\begin{proposition}\label{prop:HJM}
Under (HJM1)-(HJM3), $Q$ is an  ELMM iff
\begin{align}\label{HJM:dc}
 0 &= A(t,T) + C(t,T) + \frac{1}{2}( \parallel B(t,T) \parallel^2 +  \parallel D(t,T) \parallel^2 ), \quad \text{for }t \le T \le T^*,
\end{align}
$dQ \otimes dt - a.s.$
\end{proposition}
\begin{proof}
The formulation in terms of forward rates in \eqref{HJMsetup} directly gives that
$$ Z(t,T) :=  \frac{P(t,T)}{P(t,T^*)}
= \exp\Big(\int_T^{T^*} f(t,u) dV(u) + \int_T^{T^*} g(t,u) du \Big). $$
The dynamics of $f$, given in \eqref{HJM:f}, implies
\begin{align*}
\int_T^{T^*} f(t,u) dV(u) & = \int_T^{T^*}f(0,u) dV(u) \\
&+ \int_T^{T^*} \int_0^t a(s,u) ds \, dV(u) + \int_T^{T^*} \int_0^t b(s,u) dW_s \, dV(u) \\
&= \int_T^{T^*}f(0,u) dV(u) + \int_0^t A(s,T) ds + \int_0^t B(s,T) dW_s
\end{align*}
by the stochastic Fubini theorem (see, e.g., Theorem 6.2 in \cite{Filipovic2009}). We obtain a similar expression for the second integral. Hence, by the It\^o formula,
\begin{align*}
d Z(t,T) = Z(t,T) \Big( &A(t,T) dt +B(t,T) dW_t +  \frac{1}{2} \parallel B(t,T) \parallel ^2 dt  \\
 +&  C(t,T) dt +D(t,T) dW_t +  \frac{1}{2} \parallel D(t,T) \parallel ^2 dt \Big),
\end{align*}
for $0 \le t \le T \le T^*$.
These processes are  local martingales if and only if their drifts vanish. This is equivalent to \eqref{HJM:dc} and we conclude.
\end{proof}

\begin{remark}
The classical HJM-drift condition, i.e.\ the drift condition for the case  $f\equiv 0$, can be obtained as follows: if the limit of the  roll-overs $B(t) = \exp(\int_0^t g(s,s) ds)$ qualifies as num\'eraire, which is equivalent to the assumption that
$$ \frac{B(t) P(0,T^*)}{P(t,T^*)}, \quad 0 \le t \le T^* $$
is a true $Q$-martingale, one can change to the equivalent measure $\tilde Q$ where  $B$ is taken as num\'eraire. 
Considering the dynamics of $g$, as in \eqref{HJM:g}, under $\tilde Q$ then gives the well-known drift condition as in \cite{HJM}. 
\end{remark}

\begin{example}
Consider the simple case where $f(t,u) = \int_0^t a(s,u) du + W(t)$. We assume that $a(t,u)$ and $g(t,u)$ are deterministic functions which are bounded and continuous. Moreover, $g$ is differentiable in the first coordinate.
The initial term structure is flat, i.e.  $a(0,T)=g(0,T)=0$ for all $T \ge 0$. We have that $b(t,T)=1$, $c(t,T)=g^\prime(t,T)$ and $d(t,T)=0$.
The drift condition  \eqref{HJM:dc} in this setup reads
\begin{align} \label{HJMdrift}
0 &= \int_T^{T^*}a(t,u) dV(u)   + \int_T^{T^*} g'(t,u) du  +  \frac{1}{2} (V_{T^*} - V_T)^2.
\end{align}
We consider $(V_{T^*} - V_T)^2$ as pathwise stochastic integral and an application of  It\^o's formula
reversely in time gives
$$ (V_{T} - V_{T^*})^2 =  \int_{T^*}^T  2(V(u)- V(T^*)) \,  dV(u) +  \int_{T^*}^T  \, du. $$
We conclude that \eqref{HJM:dc} holds if and only if
$a(t,u) = (V(u) - V(T^*))$ and $g'(t,u)=1/2.$
\end{example}

\subsection{Exotic bank account processes}
\label{sec:exotic}

In this section we elaborate on a question raised in Section \ref{benchmark-example}, namely how much money can be lost by investing in the roll-over strategy. We construct an example by means of not uniformly integrable martingales in which the rollover reaches zero almost surely in finite time.

We consider a market with NAFL  and denote by $Q^*$ the measure in Theorem \ref{th1}.
Our starting point are bond prices of the form
\begin{equation} \label{bondprice:ex}
P(t,T)=E_{Q^*} \big[ \frac{N_t}{N_{T}}| \mathcal{F}_t \big]
\end{equation}
with a finite time horizon $T^*=2$ and the num\'eraire $N$, chosen as follows: let $\tau:[0,1)\to \R_{\ge 0}$ be an increasing, differentiable time transformation   with $\tau(0)=0$ and $\tau(t) \to \infty$ as $t\to 1$. The num\'eraire $N$ is given by
$$ N(t):= P(t,T^*) = \begin{cases} \exp(W_{\tau(t)}^2-\tau(t)^2) & 0 \le t <1 \\
1 & t \in [1,2].
\end{cases}
$$
with $Q^*$-Brownian motion $W$. Note that $N$ is c\`adl\`ag:  for any $\epsilon>0$
\begin{align*}
Q^*(N(t)\le \epsilon) &= Q^*( \tau(t) \xi^2 \le \log \epsilon + \tau(t)^2) \\
&= 2 \Phi(\sqrt{\tau(t) + \epsilon \tau(t)^{-1}}) -1 
\end{align*}
for a standard normal random variable $\xi$.  The last expression converges to $1$ as $\tau(t)\to \infty$ and existence of left limits of $N$ follows. However, $N$ is not uniformly integrable.
The filtration is given by $\cF_t := \sigma(W_{\tau(s)}: 0 \le s \le t),\ t \in [0,2]$, with the usual augmentation by null sets.

We  compute the bond prices with the following lemma.
\begin{lemma} \label{rem1}
For a standard normal random variable $\xi$, and $a < \half$ we have that
\begin{align*}
 E[\exp(a \xi^2 + b\xi)] &= e^{\frac{b^2}{2(1-2a)}} (1-2a)^{-\frac{1}{2}}.
 \end{align*}
 \end{lemma} \begin{proof} We start by observing that
\begin{align}\label{1}
 E[\exp(a \xi^2 + b\xi)] &= \int \frac{1}{\sqrt{2\pi}} e^{-\frac{x^2(1-2a)}{2} + bx } \, dx .
 \end{align}
Let $s:=(1-2a)^{-1}$. Then
\begin{align*}\eqref{1}
&= \int \frac{1}{\sqrt{2\pi}} e^{-\frac{x^2 -2sbx + s^2b^2 }{2s }+ \frac{sb^2}{2}}  \, dx \\
&= e^{\frac{sb^2}{2}} s^{\frac{1}{2}} \int \frac{1}{\sqrt{2 \pi s}} e^{- \frac{(x-sb)^2}{2s}}dx \\
&= e^{\frac{b^2}{2(1-2a)}} (1-2a)^{-\frac{1}{2}}. \qedhere
\end{align*}
\end{proof}

Bond prices now can be computed from \eqref{bondprice:ex}: Note that
\begin{align*}
 E_{Q^*}[ e^{-W_T^2+W_t^2} |\cF_t] &=  E_{Q^*}[e^{- (W_T-W_t)^2 - 2 W_t (W_T-W_t)}| W_t] \\
 &=E_{Q^*}[e^{- (T-t) \xi^2  - 2 W_t\sqrt{T-t} \,  \xi}| W_t] \\
 &=: \exp( W_t^2 f(t,T) - g(t,T))
\end{align*}
where $\xi$ is standard normal, independent of $W_t$; we obtain $f(t,T) = 2(T-t)\,(1+2(T-t))^{-1}$ and $g(t,T) =  \frac{1}{2} \log (1+2(T-t))$ using Lemma \ref{rem1}. Hence,
\begin{align*}
  P(t,T) = \exp( W_{\tau(t)}^2 f_\tau(t,T) - g_\tau(t,T) + \tau^2(T)-\tau^2(t) ),
\end{align*}
$0 \le t \le  T  < 1$, where we set $f_\tau(t,T):=f(\tau(t),\tau(T)))$ and similarly for $g_\tau$.
For $t \ge 1$ the term structure is flat, i.e.  $P(t,T) = 1$.

Now we turn to the limit of the roll-over account.  Fix $T<1$ and  consider $t_i^n := t_i = \tau^{-1}(iT/n)$. Then
\begin{align*}
	B^n_{t_n} = \exp( - \sum_{i=1}^n  f_\tau(t_{i-1},t_i) W_{\tau(t_{i-1})}^2  + \sum_{i=1}^n g_\tau(t_{i-1},t_i) - \tau^2(t_n)).
\end{align*}
We have that
\begin{align*}
	\exp( \sum_{i=1}^n g_\tau(t_{i-1},t_i)  )
	& =
	\exp( \half \sum_{i=1}^n \log (1+ 2( \tau(t_{i}) - \tau(t_{i-1}))))  \\
	& \to 
	e^{\tau(T)}
\end{align*}
by  Taylor expansion and continuity of $\tau$. Moreover,
\begin{align*}
	\sum_{i=1}^n  f_\tau(t_{i-1},t_i) W_{\tau(t_{i-1})}^2
	&=	
	2 \sum_{i=1}^n \frac{ W_{\tau(t_{i-1})}^2 }{1+2 (\tau(t_i)-\tau(t_{i-1}))} \, (\tau(t_i) - \tau(t_{i-1})) \\
	&\to
	2 \int_0^T W_{\tau(s)}^2 d \tau(s) = 2 \int_0^{\tau(T)} W_s^2 ds.
\end{align*}
Hence,
\begin{align*}
	B^n(T)
	&=
	\exp( - \sum_{i=1}^n W_{\tau(t_{i-1})}^2  f_\tau(t_{i-1},t_i) + \sum_{i=1}^n g_\tau(t_{i-1},t_i) - \tau^2(t_n)) \\
	& \to
	\exp( - 2 \int_0^{\tau(T)} W_s^2 ds + \tau(T) - \tau^2(T)).
\end{align*}
The discounted limit of the roll-over account turns out to be
\begin{align*}
	V(T) &= P(T,T^*)^{-1} B(T)  = \exp( -W_{\tau(T)}^2  - 2 \int_0^{\tau(T)} W_s^2 ds + \tau(T) ) \\
	&= Z(\tau(T)),
\end{align*}
letting
\begin{align}
	\label{Z}
	Z(T)  := \exp( T - 2 \int_0^T W_s^2 ds - W_T^2).
\end{align}
 We are interested in
$$ V(1) = \lim_{T \to 1} Z(\tau(T)) = \lim_{T \to \infty} Z(T). $$
The following lemma shows that $\lim_{T \to \infty} Z(T) = 0$, hence $B(1)=0$. It turns out that investing in the roll-over strategy leads to the total loss of invested money such that the classical risk-free investment strategy becomes highly risky in this example.

\begin{lemma}
Consider $Z$ as in $\eqref{Z}$. Then $Z$ converges to $0$ $Q^*$-almost surely.
\end{lemma}
\begin{proof}
Note that $Z$ is a non-negative local martingale and hence by the supermartingale convergence theorem\footnote{See Theorem 1.3.15 in \cite{KaratzasShreve}.} the limit $Z_\infty$ exists and is in $L^1$.
Moreover, we have that
$$ Z_t \le X_t :=  \exp( t - 2 \int_0^t W_s^2 ds ), \qquad t \ge 0. $$

We compute the distribution of $X_t$ by P.\ L\'evy's diagonalization procedure. Fix $T>0$. Using N. Wiener's construction of Brownian motion we obtain
$$ W_t = \sum_{k \ge 1} \frac{\sin(k\pi t /T )}{k} \xi_k \sqrt{T}, \quad t \in [0,T] $$
with i.i.d.\ standard normal $\xi_1,\xi_2,\dots$ Then
\begin{align*}
 \int_0^T W_t^2 dt &= \int_0^T \Big[ \sum_{k \ge 1} \frac{\sin(k\pi t /T )}{k} \xi_k  \sqrt{T}\Big]^2 dt \\
 &= \sum_{k,j \ge 1} \frac{T \xi_k \xi_ j}{kj} \int_0^T \sin(k\pi t /T )  \sin(j\pi t /T ) dt\\
 &= \sum_{k \ge 1} \frac{T^2 \xi_k^2}{2k^2},
\end{align*}
by orthogonality of the trigonometric functions, i.e.
$$ \int_0^T \sin(k\pi t /T )  \sin(j\pi t /T )  dt = \ind_{\{k=j\}} \frac{T}{2}. $$
Hence, for $u \ge 0$, we obtain with Lemma \ref{rem1} that
\begin{align*}
E\bigg[ e^{- u  \int_0^T W_t^2 dt } \bigg] &= \prod_{k \ge 1} E[e^{- u \frac{T^2 }{2 k^2} \xi_k^2 }] \\
&= \bigg[ \prod_{k \ge 1} (1+\frac{u T^2}{k^2})\bigg]^{-1/2} \\
&= \Big[ \frac{\sinh(\pi \sqrt{u T^2} )}{\pi \sqrt{uT^2}} \Big]^{-1/2}.
\end{align*} Next, by Fatou's lemma,
\begin{align*}
E[(X_\infty)^u ] & \le {\lim}_{T \to \infty} e^{uT}  \Big[ \frac{\sinh(\pi \sqrt{2u T^2} )}{\pi \sqrt{2uT^2}} \Big]^{-1/2}.
\end{align*}
Note that, for $0 < u <  \pi^2 /2 $,
\begin{align*}
\big[ T^{-1} e^{2uT} ( e^{\pi \sqrt{2u} T} - e^{-\pi \sqrt{2u} T}  )\big]^{-1}
	&= \frac{T} {e^{T( \pi \sqrt{2u}-2u )}- e^{-T(\pi \sqrt{2u}+2u)} } \to 0
\end{align*}
as $T\to \infty$. This shows that for the non-negative random variable $X_\infty$ and $0<u <  \pi^2 /2$, $ E[(X_\infty)^u ]=0$. By the generalized Markov inequality, for $\epsilon >0$ and $0 < u <  \pi^2 /2 $,
$$ P(X_\infty \ge \epsilon) = P((X_\infty)^u \ge \epsilon^u) \le \epsilon^{-u} E[(X_\infty)^u] = 0$$
such that $X_\infty = 0$ almost surely.
\end{proof}

\bibliographystyle{plain}

\end{document}